\documentclass[conference,a4paper,onecolumn]{IEEEtran}

\usepackage[cmex10]{amsmath}
\usepackage{amssymb,amsthm}
\usepackage{graphicx}
\usepackage{psfrag}
\usepackage{subfigure}
\usepackage{empheq}
\usepackage{mathtools}
\usepackage{import}
\usepackage{algorithm}
\usepackage{algorithmic}
\usepackage{color}
\usepackage{tcolorbox}
\usepackage{dsfont}
\usepackage{float}
\usepackage{stfloats}
\usepackage{lipsum}
\usepackage{standalone}
\usepackage{pdfpages}
\usepackage{algorithm}
\usepackage{algorithmic}
\usepackage{tcolorbox}
\usepackage{wrapfig}
\usepackage{relsize}

\DeclareMathOperator*{\argmax}{\arg\max}

\allowdisplaybreaks
\newcommand{\algrule}[1][.2pt]{\par\vskip.5\baselineskip\hrule height #1\par\vskip.5\baselineskip}
\newcommand{\RNum}[1]{\uppercase\expandafter{\romannumeral #1\relax}}

\newtheorem{lemma}{Lemma}
\newtheorem{theorem}{Theorem}

\theoremstyle{definition}

\makeatletter
\def\blfootnote{\gdef\@thefnmark{}\@footnotetext}
\makeatother

\def\cN{{\mathcal N}}

\def\cQ{{\mathcal Q}}

\def\cS{{\mathcal S}}

\def\cQ{{\mathcal Q}}

\def\cX{{\mathcal X}}
\def\cY{{\mathcal Y}}

\def\bE{{\mathbb E}}
\def\bN{{\mathbb N}}
\def\bR{{\mathbb R}}
\def\bP{{\mathbb P}}

\def\lb{\left(}
\def\rb{\right)}
\def\cardin{\left\lvert \mathcal{X} \right\lvert}

\title{Reinforcement Learning Evaluation and Solution for the Feedback Capacity of the Ising Channel with Large Alphabet}
\author{
\IEEEauthorblockN{Ziv Aharoni} 
    \IEEEauthorblockA{
    Ben-Gurion University of the Negev\\
    zivah@post.bgu.ac.il }
\and
\IEEEauthorblockN{Oron Sabag}
    \IEEEauthorblockA{
    California Institute of Technology\\
    oron@caltech.edu}
\and
\IEEEauthorblockN{Haim H. Permuter} 
    \IEEEauthorblockA{
    Ben-Gurion University of the Negev \\
    haimp@bgu.ac.il }
}

\begin{document}
\maketitle

\begin{abstract}
We propose a new method to compute the feedback capacity of unifilar finite state channels (FSCs) with memory using reinforcement learning (RL). 
The feedback capacity was previously estimated using its formulation as a Markov decision process (MDP) with dynamic programming (DP) algorithms.
However, their computational complexity grows exponentially with the channel alphabet size. 
Therefore, we use RL, and specifically its ability to parameterize value functions and policies with neural networks, to evaluate numerically the feedback capacity of channels with a large alphabet size.
The outcome of the RL algorithm is a numerical lower bound on the feedback capacity, which is used to reveal the structure of the optimal solution. The structure is modeled by a graph-based auxiliary random variable that is utilized to derive an analytic upper bound on the feedback capacity with the duality bound.
The capacity computation is concluded by verifying the tightness of the upper bound by testing whether it is BCJR invariant.
We demonstrate this method on the Ising channel with an arbitrary alphabet size. 
For an alphabet size smaller than or equal to 8, we derive the analytic solution of the capacity.
Next, the structure of the numerical solution is used to deduce a simple coding scheme that achieves the feedback capacity and serves as a lower bound for larger alphabets.
For an alphabet size greater than 8, we present an upper bound on the feedback capacity.
For an asymptotically large alphabet size, we present an asymptotic optimal coding scheme.
\end{abstract}
\blfootnote{Part of this work was presented in International Symposium on Information Theory (ISIT) 2019 \cite{aharoni2019computing}.}

\section{Introduction} \label{sec:intro}
\par The main advantage of RL is the concept function approximation, which is the key to avoiding quantization of the action and state spaces. Instead, function approximation enables the optimization of policies without visiting the entire state and action spaces.
This is the main bypass to the cardinality constraint described above, which makes the evaluation of channels with an alphabet size of $\sim100$ tractable. 
The numerical evaluation provides a numerical lower bound on the feedback capacity. Moreover, for the purpose of deriving the capacity, the numerical results form the basis for conjecturing the structure of the analytic solution.

\par The structure of the numerical solution is expressed using a directed graph, that is called a \emph{Q-graph} \cite{Q-UB}.
The Q-graph nodes and its edges represent a finite subset of the MDP states and their transitions, respectively. 
This finite subset of states forms an auxiliary RV that is used to obtain an analytic upper bound, specifically, the duality bound for unifilar FSC with feedback \cite{sabag2020duality}.
The upper bound is tight in the case where the Q-graph is BCJR invariant. That is, there exists an input distribution that visits only the states that formed the Q-graph with the same rate as the upper bound.  
Thus, the feedback capacity solution is derived.

In our work, the proposed methodology enabled us to compute the feedback capacity of the Ising channel with an alphabet size smaller than or equal to 8. Additionally, in this region, the conjectured structure enabled us to derive a capacity achieving coding scheme. 
For an alphabet size greater than 8, we provide an upper bound on the capacity.
To analyze the behaviour of the channel for an asymptotic alphabet size, we derive lower and upper bounds that are tight for an asymptotic alphabet size.

\par The remainder of the paper is organized as follows. 
Section \ref{sec:preli} includes the necessary preliminaries, and contains notation and the problem definition.
In Section \ref{sec:main-results}, we present our main results.
Section \ref{sec:rl} provides the RL algorithms applied in this work, their improvements for the feedback capacity formulation and their implementation.
In Section \ref{sec:ising}, we demonstrate the usage of RL on the Ising channel.
Section \ref{sec:conclusions} contains conclusions and a discussion of future work.

\section{Preliminaries}\label{sec:preli}
This section includes the necessary preliminaries. First, we provide notations. Second, we present the problem definition, which includes the definition of unifilar FSCs, their feedback capacity, and their formulation as an MDP. Third, we present the Q-graph and the Ising channel, on which we demonstrate our methodology. 
\subsection{Notation}
Calligraphic letters, $\mathcal{X}$, denote alphabet sets, upper-case letters, $X$, denote random variables, and lower-case letters, $x$, denote realizations.
A superscript, $x^t$, denotes the vector $(x_1,\dots,x_t)$.
The probability distribution of a random variable, $X$, is denoted by $p_X$.
We omit the subscript of the random variable when its argument has the same letter, e.g. $p(x|y)=p_{X|Y}(x|y)$.
The binary entropy is denoted by $H_2(\cdot)$ and $\log(\cdot)$ refers to the logarithm with base 2.

\subsection{Unifilar Finite State Channels}\label{app:unifilar}
A FSC is defined by the triplet $(\mathcal{X}\times\mathcal{S},p(s^\prime,y|x,s) ,\mathcal{S}\times\mathcal{Y})$, where $X$ is the channel input, $Y$ is the channel output, $S$ is the channel state at the beginning of the transmission, and $S^\prime$ is the channel state at the end of the transmission.
Also, the cardinalities $\mathcal{X},\mathcal{Y},\mathcal{S}$ are assumed to be finite.
At each time $t$, the channel has the memory-less property, that is,
\begin{equation}\label{eqn:unifilar}
  p(s_t,y_t|x^t,s^{t-1},y^{t-1}) = p(s_t|x_t,s_{t-1},y_t)p(y_t|x_t,s_{t-1}).
\end{equation}
A FSC is called \textit{unifilar} if the new channel state, $s_t$, is a time-invariant function $s_t=f(x_t,s_{t-1},y_t)$.

\subsection{Feedback Capacity of Unifilar Finite State Channels}
The feedback capacity of a unifilar FSC is presented in the following theorem.
\begin{theorem}\label{thm:capacity_unifilar}\cite[Theorem 1]{PermuterCuffVanRoyWeissman08}
The feedback capacity of a strongly connected unifilar FSC, where the initial state $s_0$ is available to both the encoder and the decoder, can be expressed by
\begin{align*}
    C_{\mathsf{FB}} &= \lim_{N\rightarrow \infty} \max_{\{p(x_t|s_{t-1},y^{t-1})\}_{t=1}^N} \frac{1}{N} \sum_{i=1}^N   I(X_i,S_{i-1};Y_i|Y^{i-1}).
\end{align*}
\end{theorem}
Note that the objective of Theorem \ref{thm:capacity_unifilar} is a multi-letter expression and, therefore, its computation is not straightforward; however, it can be computed via an MDP formulation that is given next. 

\subsection{Feedback Capacity of Unifilar Finite State Channel as Markov Decision Process}
According to \cite{PermuterCuffVanRoyWeissman08} the feedback capacity, as given in Theorem \ref{thm:capacity_unifilar}, can be formulated as an MDP. The state is the probability vector $z_{t-1} = p_{S_{t-1}|Y^{t-1}}(\cdot|y^{t-1})$, the action is the transition matrix $u_t = p_{X_t|S_{t-1},Y^{t-1}}(\cdot|\cdot,y^{t-1})$, the reward is $r_t = I(X_t,S_{t-1};Y_t|Y^{t-1}=y^{t-1})$. The next state vector at coordinate $s_t$ is given by
\begin{equation}\label{eqn:bcjr}
z_t(s_t) = \frac{\sum_{x_t,s_{t-1}}z_{t-1} 
\left(s_{t-1}\right) u_t \left(x_t, s_{t-1}\right) p(y_t|x_t,s_{t-1}) \mathds{1}[s_t=f(x_t,s_{t-1},y_t)]}
{\sum_{x_t,s_{t-1},s'_t}
z_{t-1}\left(s_{t-1}\right) u_t \left(x_ts_{t-1}\right) p(y_t|x_t,s_{t-1})\mathds{1}[s'_t=f(x_t,s_{t-1},y_t)]},
\end{equation}
where $\mathds{1}$ denotes the indicator function, $z_{t-1}(s_{t-1}) = p(s_{t-1}|y^{t-1})$ and $u_t(x_t,s_{t-1}) = p(x_t|s_{t-1},y^{t-1})$. The MDP formulation is summarized in Table \ref{tab:fb-mdp}.
\begin{table}[!h]
\caption{MDP Formulation of the Feedback capacity}
 \centering
 \begin{tabular}{|c | c|} 
 \hline
 state & $p_{S_{t-1}|Y^{t-1}}(\cdot|y^{t-1})$  \\ 
 \hline
 action & $p_{X_t|S_{t-1},Y^{t-1}}(\cdot|\cdot,y^{t-1})$  \\
 \hline
 reward & $I(X_t,S_{t-1};Y_t|Y^{t-1}=y^{t-1})$ \\
 \hline
 disturbance & $y_t$ \\
 \hline
\end{tabular}
\label{tab:fb-mdp}
\end{table}

\subsection{Q-graph}
The Q-graph \cite{Q-UB} is defined as a directed graph with edges that are labelled with symbols from the channel outputs alphabet $\cY$. By restricting the outgoing edge labels from each node to be distinct, the Q-graph can be used as a mapping of (any-length) output sequences onto the graph nodes by walking along the labelled edges. For a fixed graph, we denote the induced mapping with $\phi : \cQ \times\cY \rightarrow \cQ$,
where $\cQ$ denotes the set of graph nodes. Given a sequence of channel outputs we denote $Q_i = \Phi_i(Y^i)$, where $\Phi_i = \phi \circ \phi \circ \dots\circ\phi$ denotes the composition of $\phi$, $i$ times.

\subsection{Ising Channel}
The Ising channel model was introduced as an information theory problem by Berger and Bonomi in $1990$ \cite{Berger90IsingChannel}, 70 years after it was introduced as a problem in statistical mechanics by Lenz and his student, Ernst Ising \cite{ising1925beitrag}.
The Ising channel is a unifilar FSC and is defined by
\begin{align}
    Y &=    \begin{cases}
                X   &, \text{w.p. } 0.5   \\
                S   &, \text{w.p. } 0.5
            \end{cases},     \label{eqn:ising_out}\\
    S^\prime &= X. \label{eqn:ising_state}
\end{align}
Hence, if  $X=S$ then $Y=X=S$ w.p. 1. Otherwise, $Y$ will be one of the last two channel inputs with equal probability.
The feedback capacity of the channel was studied in \cite{Ising_artyom_IT,Ising_channel}, but here we study the Ising channel with an arbitrary alphabet, where $\cX,\cY,\cS$ are not necessarily binary.
We denote the \emph{channel cardinality} with $|\cX|$ since, by definition,  $|\cX| = |\cY| = |\cS|$.

\section{Main Results}\label{sec:main-results}
\par In this section, we present RL as a numerical tool used to estimate the feedback capacity. 
Thereafter, we present the application of RL on the Ising channel with a large alphabet to obtain the capacity, and a capacity achieving coding scheme for $|\cX|\leq 8$.
We also show an analytic upper bound on the capacity for $|\cX|>8$, and an additional coding scheme and upper bound in order to examine the channel behavior for very large alphabet sizes.

\subsection{Feedback Capacity Estimation using Reinforcement Learning}
We present RL as a numerical tool to solve the feedback capacity of unifilar FSCs using its MDP formulation. Unlike DP algorithms, RL uses neural networks (NNs) to parameterize value functions and policies, which makes it a feasible numerical tool for channels with large alphabets. The following (informal) theorem lists two algorithms for that purpose. 
\begin{theorem}[Formulation of feedback capacity as RL]\label{thm:main_RL}
    The feedback capacity and the optimal input distribution of a unifilar FSC can be estimated using two RL algorithms:
    \begin{enumerate}
        \item Deep deterministic policy gradient (DDPG).
        \item Policy optimization by unfolding (POU).
    \end{enumerate}
\end{theorem}
In the DDPG algorithm \cite{ddpg} both the value function and the policy are parameterized by NNs, while in POU only the policy is parameterized by a NN. Empirically, DDPG yielded higher numerical lower-bounds for $|\cX|\leq 15$ and POU yielded higher numerical lower-bounds for $15<|\cX| \leq 150$, and therefore both are presented. In Section \ref{sec:rl}, we present both algorithms.

\par The RL numerical results reveal bold insights into the structure of the optimal solution of the capacity problem. 
Specifically, examination of the learned input distribution showed that the visited MDP states are concentrated within a finite subset of states, and therefore can be represented by a Q-graph. 
The Q-graph is used subsequently to obtain analytic bounds on the feedback capacity, as we present next.

\subsection{Ising Channel}

In this section, we present the analytical results for the Ising channel that were deduced from RL numerical simulations, as summarized in Figure \ref{fig:res-summary}. 
The following theorem presents an application of RL to obtain the analytic feedback capacity of the Ising channel for $ \left| \cX \right|\leq 8$.
\begin{theorem}[Feedback Capacity] \label{thm:ising_cfb}
    The feedback capacity of the Ising channel with $\left| \cX \right|\leq 8$ is given by 
    \begin{equation}
        C_\mathsf{FB}(\cX) = \max_{p \in [0,1]}  2\frac{H_2(p) + (1-p) \log\left(\lvert\cX\lvert-1\right)}{p+3}.
    \end{equation}
    Equivalently, the feedback capacity can be also expressed as
    \begin{equation}
        C_\mathsf{FB}(\cX) = \frac{1}{2}\log\frac{1}{p},
    \end{equation}
    where $p$ is the unique solution of $x^4 - ((\lvert\cX\lvert-1)^4 + 4)x^3 +6x^2 -4x+1=0$ on $[0,1]$.
\end{theorem}
The proof of Theorem \ref{thm:ising_cfb} is given in Section \ref{sec:ising-small}.
We will now show a simple coding scheme that achieves the feedback capacity in Theorem \ref{thm:ising_cfb}.
Algorithm \ref{alg:code-scheme8} is applicable for any alphabet size; however, it is optimal only for $|\cX| \leq 8$ as stated in the following theorem.
\begin{theorem}[Optimal coding scheme]\label{thm:code-scheme}
The coding scheme in Algorithm \ref{alg:code-scheme8} achieves the capacity in Theorem 3 for $|\cX| \leq 8$. 
\end{theorem}
In Section \ref{sec:ising-code-small}, we prove that the coding scheme in Algorithm \ref{alg:code-scheme8} yields a zero-error code and that its maximum rate over the parameter $p$ equals the feedback capacity as given in Theorem \ref{thm:ising_cfb}. 
\par For $\lvert \cX \lvert > 8$, the structure of the analytic solution changes. Unlike the solution for $\left| \cX \right|\leq 8$, the Q-graph induced by the numerical results cannot be described with a finite set of nodes. 
Nevertheless, the numerical results dictate a sub-optimal structure that induces an upper bound for $\left| \cX \right| > 8$. The upper bound for $\left| \cX \right| > 8$ is shown in the following theorem.

\begin{figure}[!t]
    \centering
    \includegraphics[width=0.6\linewidth]{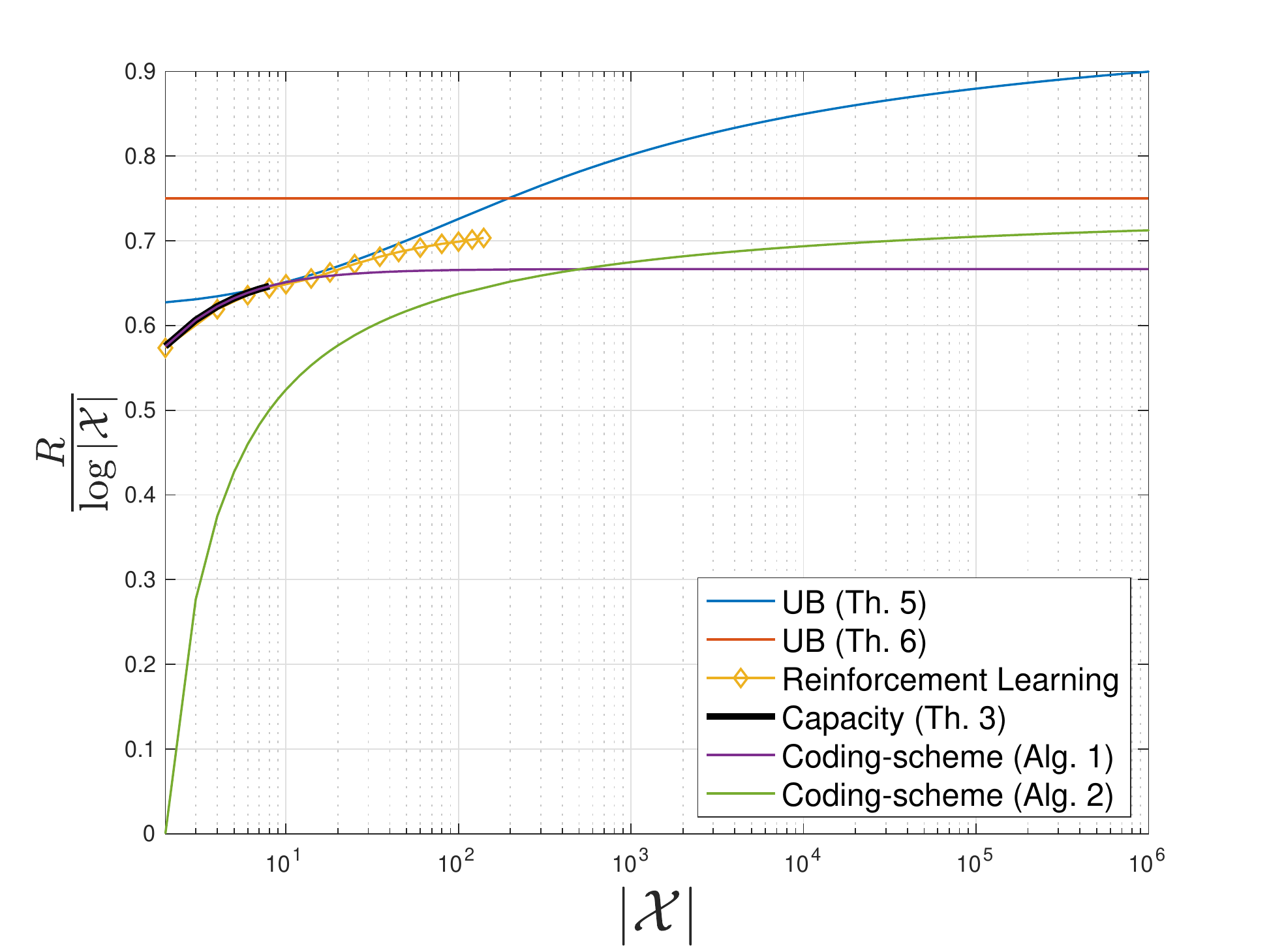}
    \caption{Summary of the analytic bounds and the numerical results obtained by the RL simulations for varying alphabet sizes. The rates/bounds are normalized by $\log|\cX|$.}
    \label{fig:res-summary}
\end{figure}

\begin{theorem}[Upper bound for $|\cX|>8$]\label{thm:ub-large_cardinality}
The feedback capacity of the Ising channel satisfies
        \begin{equation*}
            \mathsf{C}_{fb}\left(\cX\right)\leq \frac{1}{2}\log\frac{\left|\mathcal{X}\right|}{p},
        \end{equation*}
        where $p$ is the unique root of $x^2-\left(2+\frac{\left(\left|\mathcal{X}\right|-1\right)^2}{16\left|\mathcal{X}\right|}\right)x+1$ in $[0,1]$.
\end{theorem}

\par For any alphabet size, an upper bound on the capacity and a coding scheme are presented in the following theorem.
\begin{theorem}[Asymptotic performance]\label{thm:ub-general}
For any alphabet size $\lvert\cX\lvert > 2$, the feedback capacity of the Ising channel satisfies
\begin{equation}
    C_\mathsf{FB}(\cX) \leq\frac{3}{4}\log\lvert\cX\lvert.
\end{equation}
Also, for any alphabet size $\lvert\cX\lvert > 2$, there is a simple coding scheme with the following rate:
\begin{equation}
    R(\cX) = \frac{3}{4}\log\frac{\lvert\cX\lvert}{2}.
\end{equation}
Therefore, $\frac{3}{4}\log \frac{\lvert\cX\lvert}{2} \leq C_\mathsf{FB}(\cX) \leq \frac{3}{4}\log\lvert\cX\lvert$.
\end{theorem}
\noindent Theorem \ref{thm:ub-general} is proved in Section \ref{sec:ising-asymp}.

The analytical results and the numerical results of the RL algorithms are summarized in Figure \ref{fig:res-summary}. 
The RL simulation is the yellow curve. We simulated the RL algorithms up to a size of $150$ due to a computational memory constraint. The bold-black curve illustrates the analytical capacity that appears in Theorem \ref{thm:ising_cfb} for $|\mathcal X|\le 8$. One can see that the capacity achieving coding scheme (in purple) coincides with the RL simulation for $|\cX|\leq 8$. However, it converges to $\frac{2}{3}\log |\cX|$ for large alphabets, while RL continues to improve. To back up this observation, our improved lower and upper bounds for the asymptotic case (green and orange curves, respectively) are shown. For large $|\cX|$, both converge to $\frac{3}{4}\log|\cX|$ with a constant difference of $\frac{3}{4}$. We also present the upper bound from Theorem \ref{thm:ub-large_cardinality}, which outperforms the others for $8 < |\cX|\le 200$.
\begin{algorithm}[H]
    \caption{Capacity achieving coding scheme for $|\cX|\leq 8$}
    \label{alg:code-scheme8}
    \textbf{Code construction and initialization:}
    \begin{itemize}
        \item[-] Transform the $n$ uniform bits of the message into a stream of symbols (from $\mathcal X$) with the following statistics:
        \begin{equation}
            \nu_i = \begin{cases}
                        \nu_{i-1} &, \text{w.p. } p \\
                        \text{Unif}[\cX \backslash \{\nu_{i-1}\}] &, \text{w.p. } 1-p,
                    \end{cases}
        \end{equation}
        with $\nu_0=0$. The mapping can be done using enumerative coding \cite{1054929}
        \item[-] Transmit a symbol twice to set the initial state of the channel $s_0$
    \end{itemize}
    \algrule
    \textbf{Encoder:}
    \begin{algorithmic}
    \STATE Transmit $\nu_t$ and observe $y_t$
    \IF{$y_t = s_{t-1}$}
        \STATE Re-transmit $\nu_t$
    \ENDIF
    \end{algorithmic}
    \algrule
    \textbf{Decoder:}
    \begin{algorithmic}
    \STATE Receive $y_t$
    \IF{$y_t \neq y_{t-1}$}
        \STATE Store $y_t$ as an information symbol 
    \ELSE
        \STATE Ignore $y_t$ and store $y_{t+1}$ as a new information symbol 
    \ENDIF
    \end{algorithmic}
\end{algorithm}

\section{Formulating the Feedback Capacity as Reinforcement Learning}\label{sec:rl}
\par In this section, we give a brief background on RL, based on \cite{sutton2018reinforcement}. Then, we formulate the feedback capacity of a unifilar FSC as an RL problem and provide the algorithms to compute the capacity. 
An important benefit of the formulation is that the RL environment is completely known, unlike the general assumption in classic RL.
Therefore, we leverage the full knowledge of the channel equations and use two algorithms.
The first is the DDPG algorithm with improvements; these are based on the knowledge of the environment, and on a prior assumption that the optimal solution has a structure.
The second algorithm is POU that uses the knowledge of the environment to optimize the feedback capacity directly.
The DDPG algorithm estimates both the value function and the policy, and therefore its results are easier to interpret. However, it yielded better lower bounds (compared with POU) only for $|\cX|\leq 15$, and did not converge for alphabets beyond $|\cX|=15$. The POU algorithm, which only estimates the policy, performed better for $|\cX|>15$ empirically, but was less accurate for $|\cX|\le 15$.

\subsection{RL Setting}
\par The RL setting comprises an agent that interacts with a state-dependent environment whose input is an action, and the output is a state and a reward. 
Formally, at time $t$, the environment state is $z_{t-1}$, and an action $u_t\in\mathcal{U}$ is chosen by the agent. Then, a reward $r_t\in\mathcal{R}$ and a new state $z_t\in\mathcal{Z}$ are generated by the environment, and are made available to the agent (Figure \ref{fig:rl-formulation}). 
The environment is assumed to satisfy the Markov property 
\begin{equation} 
    p\left(r_t, z_t \lvert z^{t-1},u^t,r^{t-1}\right) = p\left(r_t,z_t \lvert z_{t-1},u_t\right), \label{eqn:markov}
\end{equation}
and hence, it can be characterized by the time-invariant distribution $p\left(r_t,z_t \lvert z_{t-1},u_t\right)$ only. The agent's \emph{policy} is defined as the sequence of actions $\pi = \{u_1, u_2, \dots \}$.
\begin{figure}[t]
    \centering
    \subfigure[General RL setting]
    {
    \includegraphics[width=0.29\linewidth]{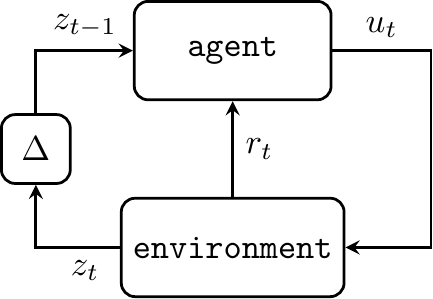}
    }
    \qquad\qquad
    \subfigure[Feedback capacity formulated in the RL setting]
    {
     \includegraphics[width=0.45\linewidth]{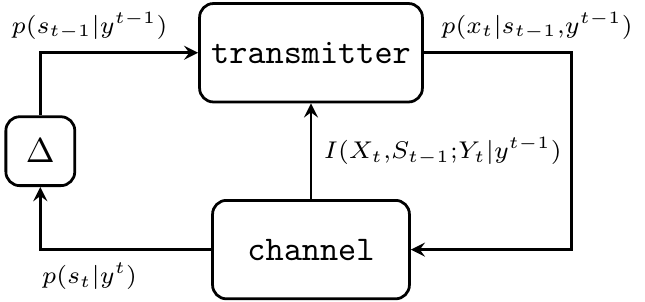}
    }
  \caption{A description of (a) the general RL setting and (b) the feedback capacity problem formulated in the RL setting.}
  \label{fig:rl-formulation}
\end{figure}
\begin{figure}[b]
    \centering
    \includegraphics[scale=1.2, trim=0 0 0 0, clip]{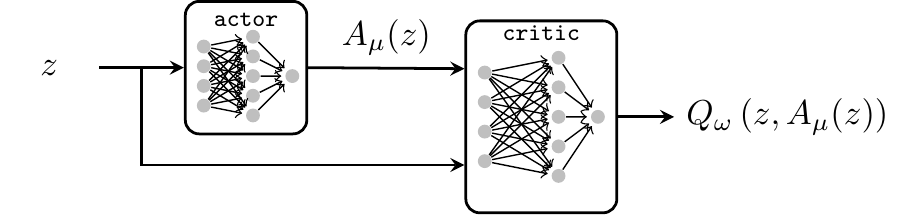}
    \caption{Depiction of actor and critic networks. The actor network comprises a NN that maps the state $z$ to an action $A_\mu(z)$. The critic NN maps the tuple $(z, u)$ to an estimate of expected future cumulative rewards.}
    \label{fig:actor-critic}
\end{figure}

\par The objective of the agent is to choose a policy that yields maximal accumulated rewards across a predetermined horizon $h\in\bN$. Here, we consider an \emph{infinite-horizon average-reward} setting, where the agent-environment interaction lasts forever, and the goal of the agent is to maximize the average reward gained during the interaction.
The average reward of the agent is defined by 
\begin{equation}\label{eqn:average_reward}
    \rho(\pi) = \lim_{h\rightarrow\infty} \frac{1}{h}\sum_{t=1}^h \bE_\pi[R_t|Z_0],
\end{equation}
where the rewards depend on the initial state $Z_0$ and on the actions taken according the policy $\pi$.

\par The \emph{differential return} of the agent is defined by 
\begin{equation}
    G_t = R_t -\rho(\pi) + R_{t+1} -\rho(\pi) + R_{t+1} -\rho(\pi) + \cdots.
\end{equation}
Accordingly, the state-action value function $Q_\pi(z,u)$  is defined as 
\begin{equation}
    Q_\pi(z,u) = \bE_\pi \left[ {G_t \lvert Z_{t-1}=z, U_t=u} \right]. \label{eqn:Q}
\end{equation}
That is, the expected rewards for taking action $u$ at state $z$ and thereafter following policy $\pi$.
Using the Markov property \eqref{eqn:markov} of the environment, one can write \eqref{eqn:Q} as the sum of the immediate and future rewards, i.e.,
\begin{align}\
    Q_\pi(z,u) = &\bE \left[ R_t \lvert Z_{t-1}=z, U_t=u \right]- \rho(\pi) + 
                 \bE_\pi \left[ Q_\pi(Z_t, U_{t+1})\lvert Z_{t-1}=z, U_t=u\right], \label{eqn:Q_decompose}
\end{align}
that is the Bellman equation \cite{value_iteration_bellman}, which is essential for estimating the function $Q_{\pi}$.
Given an estimation of the state-action value, it forms the basis for the improvement of a given policy.
That is, for each state $z \in \mathcal{Z}$, the current action $\pi(z)$ can be improved to the action $\pi^\prime(z)$ by choosing
\begin{equation}
    \pi^\prime(z) = \argmax_{u} Q_\pi(z, u).
\end{equation}

\par The \textit{function approximators} in RL are parameterized models for $Q_\pi(z,u), \pi(z)$.
The \textit{actor} is defined by $A_\mu(z)$, a parametric model of $\pi(z)$, whose parameters are $\mu$.
The \textit{critic} is defined by $Q_\omega(z,u)$, a parametric model with parameters $\omega$ of the state-action value function that corresponds to the policy $A_\mu(z)$.
Generally, in deep RL, the actor and critic are modeled by NNs, as shown in Fig \ref{fig:actor-critic}. 

\subsection{Formulation of the Capacity as an RL}
The MDP formulation \cite{PermuterCuffVanRoyWeissman08} of the feedback capacity is used to convert the multi-letter capacity formula in Theorem \ref{thm:capacity_unifilar} into an RL setting. 
The formulation is depicted in Figure \ref{fig:rl-formulation}. Under this formulation the state $z_{t-1} = p_{S_{t-1}|Y^{t-1}}(\cdot | y^{t-1})$ is the probability vector of the channel state given the channel outputs feedback. The action $u_t = p_{X_t|S_{t-1},Z_{t-1}=z_{t-1}}$ is the conditional probability of the channel input conditioned on the channel state. The reward is $r_t = I(X_t,S_{t-1};Y_t|z_{t-1}, u_t)$ and the next state is given by the evolution of $z_t$ and is described in \eqref{eqn:bcjr}. An equivalent notation denotes $r_t = g(z_{t-1}, u_t), z_t = f(z_{t-1}, u_t, w_t)$, where $g,f$ are the reward and next state function, respectively. The disturbance, $w_t$, is chosen as the channel output $y_t$.

\subsection{Deep Deterministic Policy Gradient (DDPG) Algorithm}
\par In this section we elaborate on the implementation of the DDPG \cite{ddpg}, including the necessary adjustments to the feedback capacity formulation.

\subsubsection{Algorithm}
The DDPG algorithm \cite{ddpg} is a deep RL algorithm for deterministic policies and continuous state and action spaces, as needed by the feedback capacity underlying MDP. 
The training procedure comprises $M$ episodes, where each episode contains $T$ sequential steps. 
A single step of the algorithm comprises two parallel \emph{operations}: (1) collecting experience from the environment, and (2) improving the actor and critic networks performance by training them using the accumulated data.

\begin{figure}[!ht]
    \centering
    \includegraphics[scale=0.8, trim=0 0 0 0, clip]{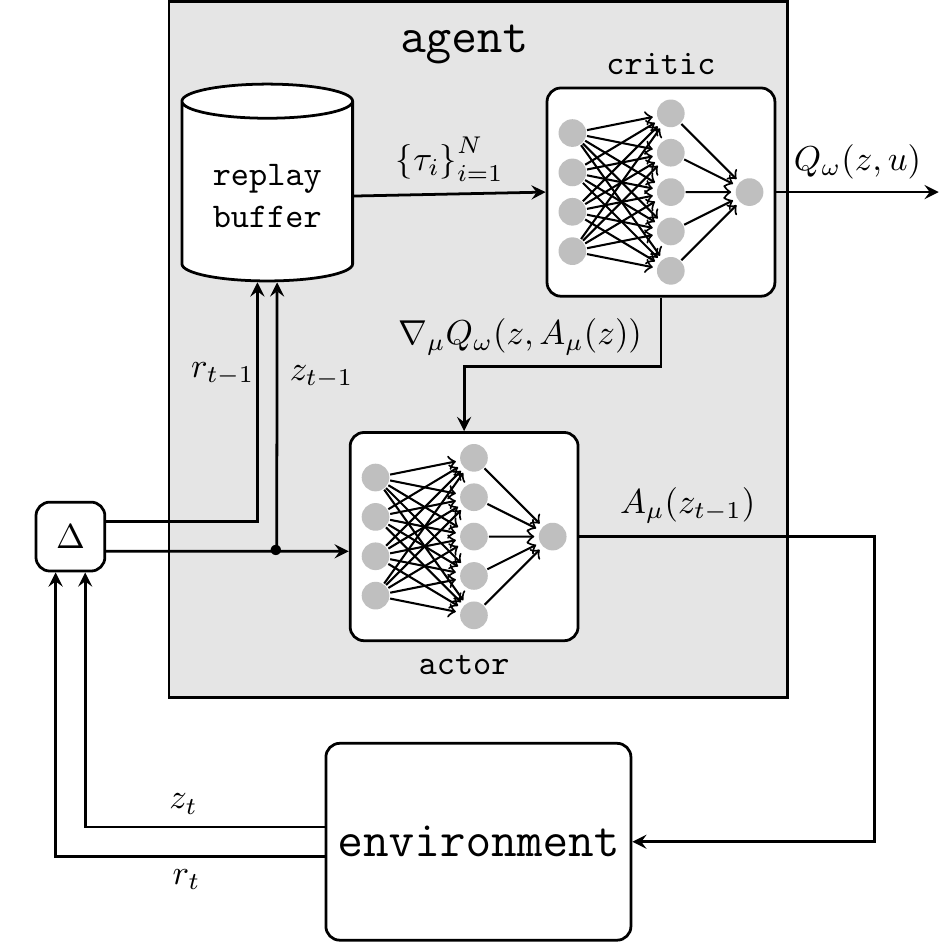}
    \caption{Depiction of the work flow of the DDPG algorithm. At each time step $t$, the agent samples a transition from the environment using an $\epsilon$-greedy policy and stores the transition in the replay buffer. Simultaneously, $N$ past transitions $\left\{\tau_i\right\}_{i=1}^N$ are drawn from the replay buffer and used to update the critic and actor NN according to \eqref{eqn:critic-update} and \eqref{eqn:actor-update}, respectively. }
    \label{fig:ddpg-workflow}
\end{figure}

\par In the first operation, the agent collects experience from the environment. 
Given the current state $z_{t-1}$, the agent chooses an action $u_t$ according to an exploration policy. Here, the action is a probability distribution, and therefore exploration is applied by adding noise to the actor network's last hidden layer, and not by adding noise to the network output as done in \cite{ddpg}. We denote a noisy action at state $z_{t-1}$ by $A_\mu(z_{t-1}; N_t)$, where $\{N_t\}$ is an i.i.d. Gaussian process with $N_t \sim \cN(0,\sigma^2)$.
After taking the action $A_\mu(z_{t-1};N_t)$, the agent observes the incurred reward $r_t$ and the next state $z_t$. 
Subsequently, the \emph{transition} tuple 
$$\tau = \left(z_{t-1},u_t,r_t,z_t\right)$$ 
is stored in a \textit{replay buffer}, a bank of experience, that is used to improve the actor and critic networks in the second operation.

\begin{algorithm}[h]
  \caption{DDPG algorithm for feedback capacity of unifilar FSC \label{alg:ddpg}}
  \begin{algorithmic}
    \STATE Initialize the critic $Q_\omega$ and actor
    $A_\mu$ networks with random weights $\omega$ and $\mu$, respectively
    \STATE Initialize target networks $Q^\prime_\omega$ and $A^\prime_\mu$ with weights $\omega^\prime 
    \leftarrow \omega$, $\mu^\prime \leftarrow \mu$
    \STATE Initialize an empty replay buffer $R$
    \STATE Initialize moving average parameter $\alpha$
    \FOR{episode = 1:M}
      \STATE Initialize a random process $\{N_t\}$ for action
      exploration
       
      \STATE Set $\rho_{MC} = \frac{1}{T_{MC}}\sum_{t=0}^{T_{MC}-1} r_{t+1}$ by a Monte-Carlo evaluation of the average reward of $A_\mu$

      \STATE Randomize initial state $z_0$ from the $|\cX|$-simplex
      \FOR{step = 1:T}
        \STATE Select noisy action $u_t = A_\mu(z_{t-1} ; N_t) $
        \STATE Execute action $u_t$ and observe $(r_t,z_t)$
        \STATE Store transition $(z_{t-1}, u_t,
                r_t, z_t)$ in $R$
        \STATE Sample a random batch of $N$ transitions
               $\left\{(z_{i-1}, u_i,
        r_i, z_i)\right\}_{i=1}^N$ from $R$
        \STATE Set $b_i = r_i - \rho_{MC} + Q^\prime_\omega\left(z_i, A^\prime_\mu(z_i)\right)$ \label{alg-step:td-target}

        \STATE Update critic by minimizing the loss:
               $L\left( \omega \right) = \frac{1}{N} \sum_{i=1}^{N
               } \left[Q_\omega\left(z_{i-1}, A_\mu(z_{i-1})\right) - b_i \right]^2$
        \STATE Update the actor policy using the sampled policy gradient:
        \begin{equation*}
            \frac{1}{N} \sum_{i=1}^{N} \nabla_a Q_\omega\left(z_{i-1}, a\right)\lvert_{a=A_\mu(z_{i-1})} \nabla_\mu A_\mu\left( z_{i-1}\right)
         \end{equation*}
        \STATE Update the target networks:
          \begin{align*}
            \omega^\prime &\leftarrow \alpha \omega + (1 - \alpha) \omega^\prime \\
            \mu^\prime &\leftarrow \alpha \mu +
                (1 - \alpha) \mu^\prime
          \end{align*}

        \ENDFOR
    \ENDFOR
    \RETURN $\rho_{MC}= \frac{1}{T_{MC}}\sum_{t=0}^{T_{MC}-1} r_{t+1}$
  \end{algorithmic}
\end{algorithm}

\par The second operation entails training the actor and critic networks.
First, $N$ transitions $\left\{ \tau_i \right\}_{i=1}^N$ are drawn uniformly from the replay buffer. 
Second, for each transition, the target $b_i$ is computed based on the right-hand-side of \eqref{eqn:Q_decompose}:
\begin{equation}\label{eqn:rl_sample_target}
    b_i = r_i - \rho_{MC} +  Q^\prime_\omega\left(z_i, A^\prime_\mu(z_i)\right), \quad i=1,\dots,N.
\end{equation}
The target is the sampled estimate of future rewards; for numerical reasons it is computed using a moving average of $Q_\omega,A_\mu$, which are the target networks, $Q^\prime_\omega,A^\prime_\mu$.
The term $\rho_{MC}$ is the estimate of the average reward, which is updated at the beginning of every episode by a Monte-Carlo evaluation of $T_{MC}$ steps by $\frac{1}{T_{MC}}\sum_{t=0}^{T_{MC}-1} r_{t+1}$.
Then, we minimize the following objective with respect to the parameters of the critic network $\omega$ as given by
\begin{equation}\label{eqn:critic-update}
    L\left( \omega \right) = \frac{1}{N} \sum_{i=1}^{N} \left[Q_\omega\left(z_{i-1}, A_\mu(z_{i-1})\right) - b_i \right]^2.
\end{equation}
The aim of this update is to train the critic to comply with the Bellman equation \eqref{eqn:Q_decompose}.
Afterwards, we train the actor to maximize the critic's estimation of future cumulative rewards. 
That is, we train the actor to choose actions that result in high cumulative rewards according to the critic's estimation.
The formula for the actor update is given by
\begin{equation} \label{eqn:actor-update}
    \frac{1}{N} \sum_{i=1}^{N} \nabla_a Q_\omega\left(z_{i-1}, a\right)\lvert_{a=A_\mu(z_{i-1})} \nabla_\mu A_\mu\left( z_{i-1}\right).
\end{equation}
Finally, the agent updates its current state to be $z_t$ and moves to the next time step. 
\par To conclude, the algorithm alternates between improving the critic's estimation of future cumulative rewards and training the actor to choose actions that maximize the critic's estimation.
The algorithm is given in Algorithm \ref{alg:ddpg} and its workflow is depicted in Figure \ref{fig:ddpg-workflow}.

\subsubsection{Improvements}
We propose two improvements for the DDPG algorithm.
The first improvement uses the knowledge of the environment to reduce the variance of the estimation of $Q_\pi$ by replacing samples with expectations. Instead of calculating the right-hand-side of \eqref{eqn:Q_decompose} as done in  \eqref{eqn:rl_sample_target}, we compute the expectation over all possible next states by
\begin{align}\label{eqn:rl_exp_target}
    b_i &= r_i - \rho_{MC}  + \sum_{w \in \cY} p(w|z_{i-1},u_i)Q^\prime_\omega\left(z_i, A^\prime_\mu(z_i)\right),
\end{align}
where $z_i=f(z_{i-1},u_i,w)$. This is possible since the disturbance (the channel output) has finite cardinality. 

\par The second improvement is a variant of importance sampling \cite{importance_sampling}. This is essential since there are states that are visited rarely, and in the current technique are rarely used to improve the policy. For this purpose, we modify the replay buffer to store transitions as clusters. 
Each time a new transition arrives at the buffer, its max-norm distance with all cluster centers is calculated.
The distance from the closest cluster is compared with a threshold (typically $\sim0.1$).
In the case where the distance is smaller than the threshold, the transition is stored in the corresponding cluster; else, a new cluster is added with the new transition. 
For sampling, instead of drawing transitions uniformly over the entire buffer, we first sample uniformly from the clusters, and then sample uniformly from within the sampled cluster. This modification increases the probability that \emph{rare states} will be drawn from the replay buffer. Therefore, the value function estimation improves in rare states, which consequently yields better policies in rare states.

\subsubsection{Implementation}\label{sec:ddpg-implement}
\par We model $Q_\omega(z,u)$, $A_\mu(z)$ with two NNs, each of which is composed of three fully connected hidden layers of 300 units separated by a batch normalization layer.
The actor network input is the state $z$ and its output is a matrix $A_\mu(z) \in \bR^{|\cS|\times |\cX|}$ such that $A_\mu(z) \mathbf{1} = \mathbf{1}$. 
The critic network input is the tuple $\left( z, A_\mu(z)\right)$ and its output is a scalar, which is the estimate for the cumulative future rewards.
In our experiments, we trained the networks for $M=10^4$ episodes. 
Each episode length is $T=500$ steps. The Monte-Carlo evaluation length of average reward is $T_{MC} = 10^8$.
For the exploration, we added Gaussian noise with zero mean and variance $\sigma^2=0.05$ to the last layer of the actor network.
The implementation details are published in github\footnote{https://github.com/zivaharoni/capacity-rl}.

\subsection{Policy Optimization by Unfolding (POU) Algorithm}
The POU algorithm utilizes the knowledge of the RL environment to optimize the policy without estimating the value function. That is, we optimize the average of consecutive rewards directly. This is done by using the reward function $g$ and the next state function $f$ to define a mapping between an initial MDP state and the average of the consecutive $n$ rewards. The mapping is finally used as an objective to optimize the policy.

\begin{figure}[!b]
    \centering
    \includegraphics[scale=1.0, trim=0 0 0 0, clip]{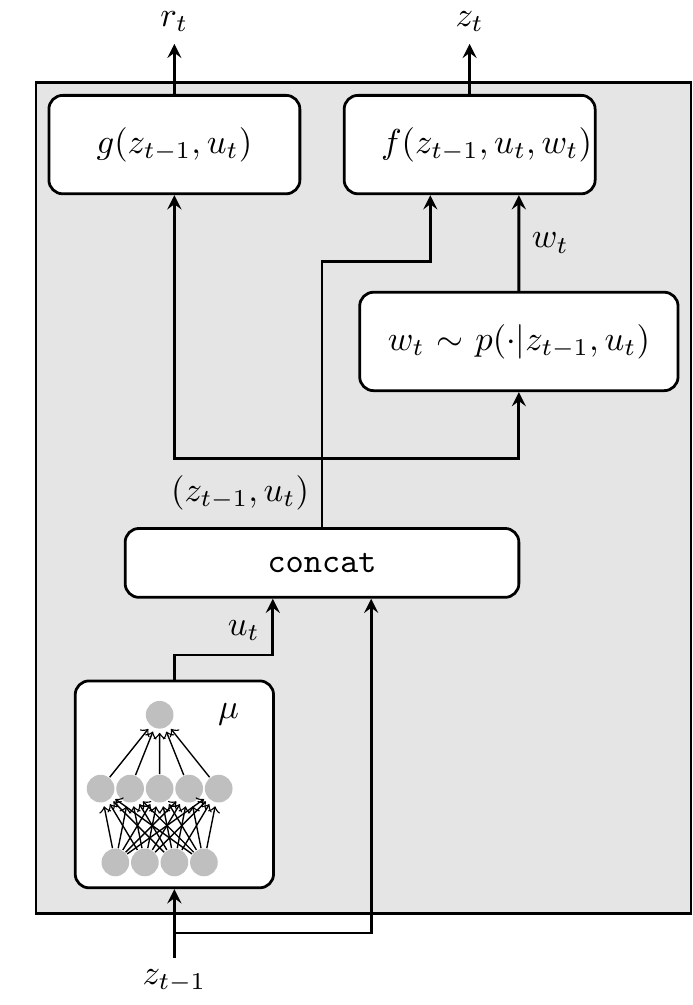}
    \caption{A single step of the environment. The input of the block is the current RL state $z_{t-1}$ and the outputs are the immediate reward $r_t$ and the next sampled state $z_t$. Initially, the block uses the actor to construct the tuple $(z_{t-1},u_t)$. Afterwards, it samples the disturbance from $w_t \sim p(\cdot|z_{t-1}, u_t)$, and finally, uses $g$ and $f$ to compute the reward and the next state, respectively.}
    \label{fig:policy-opt-cell}
\end{figure}

\subsubsection{Algorithm}
\par Let us denote the policy-dependent reward function by
\begin{align}
    R_\mu(z)&= g\big(z,\mu(z)\big),
\end{align}
that depends exclusively on $z$ since the policy $\mu$ is a deterministic policy.
Consequently, we define the average reward over $n$ consecutive time steps for an initial MDP state $z_0$ by
\begin{align}\label{eq:POU_accum_reward}
    R^n_\mu(z_0)&= \frac{1}{n}\sum_{t=1}^n \mathbb{E} \big[ R_\mu(Z_{t-1}) \big] \nonumber\\
    &= \frac{1}{n}\left[R_\mu(z_0) + \sum_{t=2}^n\mathbb{E} \big[ R_\mu(Z_{t-1}) \big]\right],
\end{align}
where $Z_t = f(Z_{t-1}, \mu(Z_{t-1}), W_t)$ and hence, the expectation is implicitly taken with respect to $P_{W_{t}|Z_{t-1},\mu(Z_{t-1})}$. That is since the disturbance $W_t$ is conditionally independent of the past given the previous MDP state and the action $\big(Z_{t-1},\mu(Z_{t-1})\big)$.
\begin{figure}[!t]
    \centering
    \includegraphics[scale=0.9, trim=10pt 0 0 0, clip]{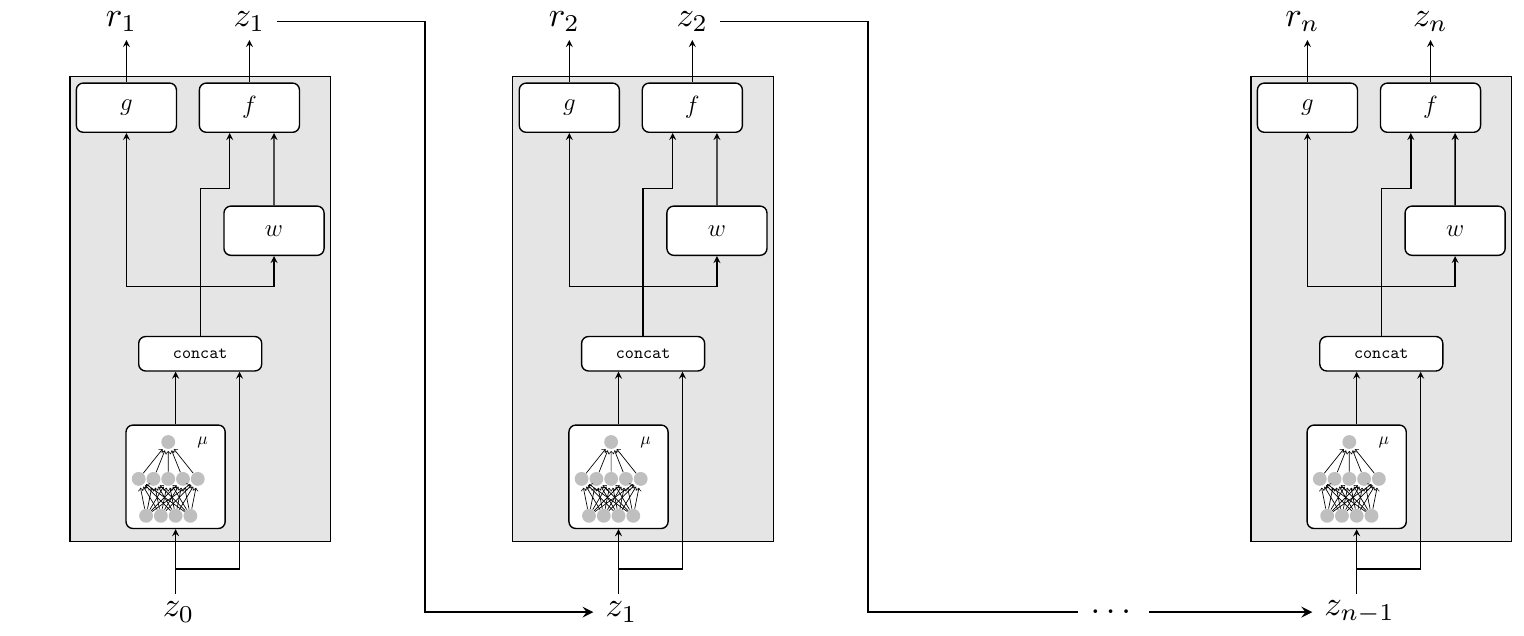}
    \caption{The interaction with the channel unrolled across subsequent time steps. The weights of the actor network are shared across time steps.}
    \label{fig:policy-opt}
\end{figure}

\par The choice of the interaction length parameter $n$ affects directly the performance of the optimized policy. Specifically, as $n$ increases the policy is optimized over more rewards in future steps rather than immediate rewards. For instance, choosing $n=1$ translates to optimizing the immediate reward, which consequently yields a greedy policy. As shown in Figure \ref{fig:res-summary}, an interaction over relatively small $n$, e.g., $n\sim20$, is sufficient to achieve policies with long-term high performance. However, the number of possible MDP states over an interaction of $n$ steps grows exponentially as $|\cY|^n$ (recall the disturbance in our case is the channel output $Y$).

\par To resolve this practical issue, the POU algorithm proposes a simple, yet efficient, method to \emph{unfold} the interaction with the environment. Given a policy $\mu$ and an initial state $z_0$, we sample $n$ MDP states and rewards consecutively according the following law:
\begin{align} \label{eqn:pou_single_time_step}
    r_t &= R_\mu(z_{t-1}), \nonumber\\
    W_{t}&\sim P_{W|Z,U}(\cdot |Z=z_{t-1},U = \mu(z_{t-1})),\nonumber\\
    z_{t}&= f(z_{t-1},\mu(z_{t-1}),w_{t}),
\end{align}
where the disturbance $W_{t}$ is sampled conditioned on the previous MDP state and the action $\mu(Z_{t-1})$. Note that this law is dictated by the RL environment and the chosen policy and is not subject to the planning horizon $n$.
For a single $t$, the law in \eqref{eqn:pou_single_time_step} describes a single step where the agent interacts with the environment, as shown in Figure \ref{fig:policy-opt-cell}. The interaction with the environment for $n$ consecutive steps is shown in Figure \ref{fig:policy-opt}.
\par After applying \eqref{eqn:pou_single_time_step} $n$ times, the disturbance sequence $(w_1,\dots,w_{n-1}, w_n)$ is sampled, and subsequently, a deterministic, differentiable mapping between $z_0$ and the average reward is established. Specifically, we can compute the derivative of the average reward in \eqref{eq:POU_accum_reward} without the expectation, that is, 
\begin{align}
    \nabla_\mu \left[ \frac{1}{n}\sum_{t=1}^n  R_\mu(z_{t-1}) \right].
\end{align}
Then, we update the policy $\mu$ with the standard gradient ascent update as:
\begin{align}
    \mu = \mu + \eta \nabla_\mu \left[ \frac{1}{n}\sum_{t=1}^n  R_\mu(z_{t-1}) \right],
\end{align}
where $\eta$ is the step size.
\begin{algorithm}[!ht]
  \caption{POU algorithm for feedback capacity of unifilar FSC \label{alg:pou}}
  \begin{algorithmic}
    \STATE Initialize actor $A_\mu$ with random weights $\mu$
    \STATE Initialize learning rate $\eta$.
    \FOR{episode = 1:M}
      \STATE Sample $z_0$ uniformly from $(|\cX|-1)$-simplex
      \FOR{$t = 1:T$}
        \STATE Conditioned on $z_{0}$ and $A_\mu$, sample $(w_1,\dots,w_{n})$ according to \eqref{eqn:pou_single_time_step}
        \STATE Compute average of $n$ rewards $\frac{1}{n}\sum_{i=1}^{n}  R_\mu(z_{i-1})$ 
        \STATE Update the actor parameters using gradient ascent
        \begin{equation*}
            \mu = \mu + \eta \nabla_\mu \left[ \frac{1}{n}\sum_{i=1}^n  R_\mu(z_{i-1}) \right]
         \end{equation*}
        \STATE Update the initial state $z_0 = z_n$

        \ENDFOR
    \ENDFOR
    \RETURN $\rho_{MC}= \frac{1}{T_{MC}}\sum_{t=0}^{T_{MC}-1} r_{t+1}$

  \end{algorithmic}
\end{algorithm}
This procedure is repeated using the last state $z_n$ as the initial state of the next consecutive $n$ steps. This is shown in Algorithm \ref{alg:pou}. 

\subsubsection{Implementation}
\par The actor network is implemented exactly as described in Section \ref{sec:ddpg-implement}.
For training, we trained the actor network for $M=10^3$ episodes, each containing $T=10^2$ consecutive $n$-blocks. 
Each block was chosen to have length $n=20$. The Monte-Carlo evaluation length of average reward is $T_{MC} = 10^8$. For exploration, we used dropout \cite{dropout} on the actor network throughout training. The implementation details are published in github\footnote{https://github.com/zivaharoni/capacity-rl-po}.

\section{The Ising Channel}\label{sec:ising}
\par This section demonstrates the usage of RL to obtain analytic results on the Ising channel.
First, we describe a methodology to convert the numerical results into analytic results. 
Then, we demonstrate the implementation on the Ising channel.

\subsection{Converting the Numerical Results into Analytic Bounds}\label{sec:analytic-method}
In this section, we describe the conversion of numerical results into analytic results; specifically, we demonstrate this method on the Ising channel with $|\cX| = 3$.
First, we describe how to extract the structure of the numerical solution. 
Afterwards, we present how to use the structure to obtain an analytic upper bound and how to verify whether this bound is tight. 
\begin{figure}[t]
    \centering
    \subfigure[State histogram in training]
    {
    \includegraphics[width=0.45\linewidth]{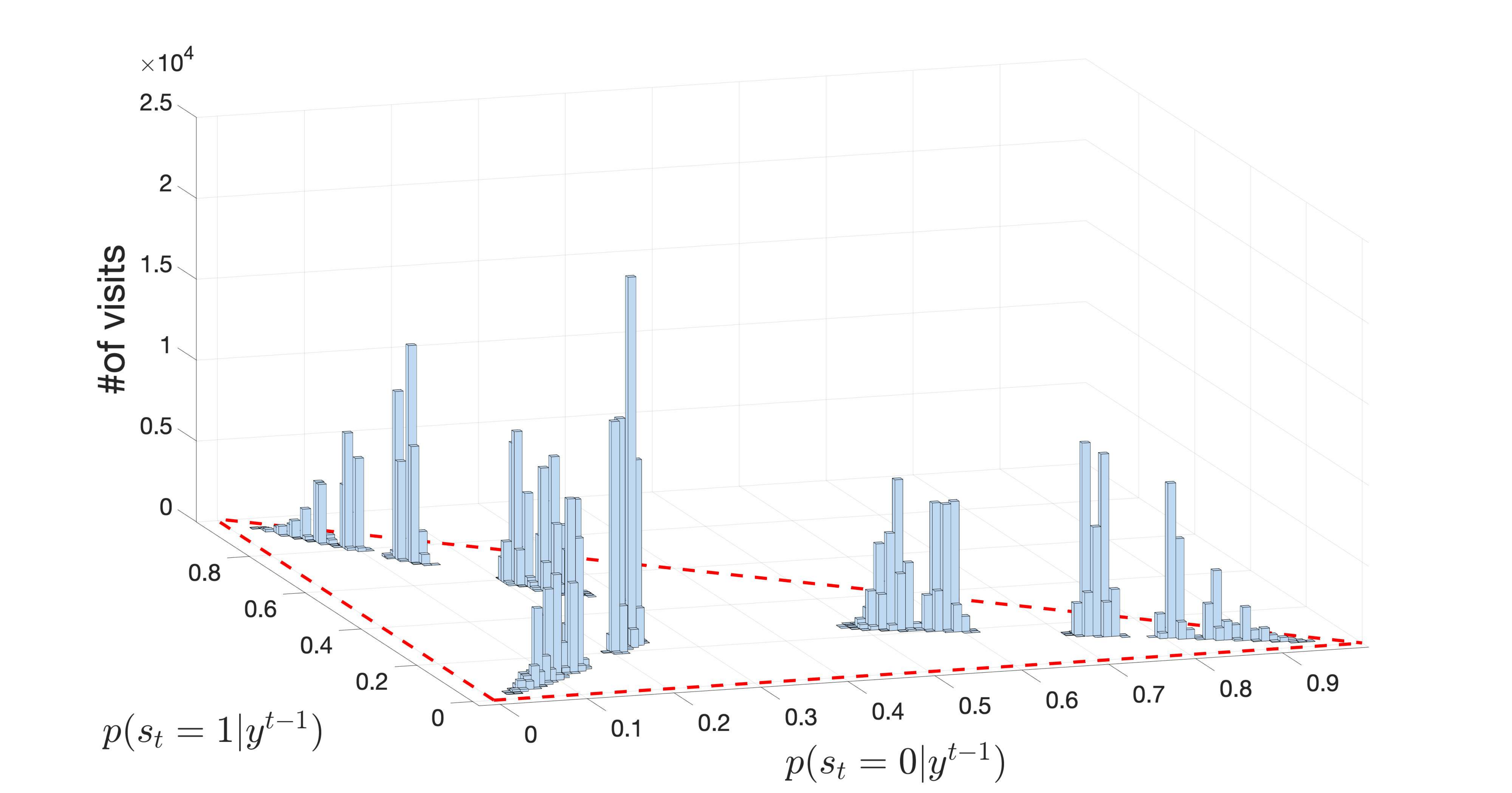}
    }
    \qquad
    \subfigure[Final state histogram]
    {
    \includegraphics[width=0.45\linewidth]{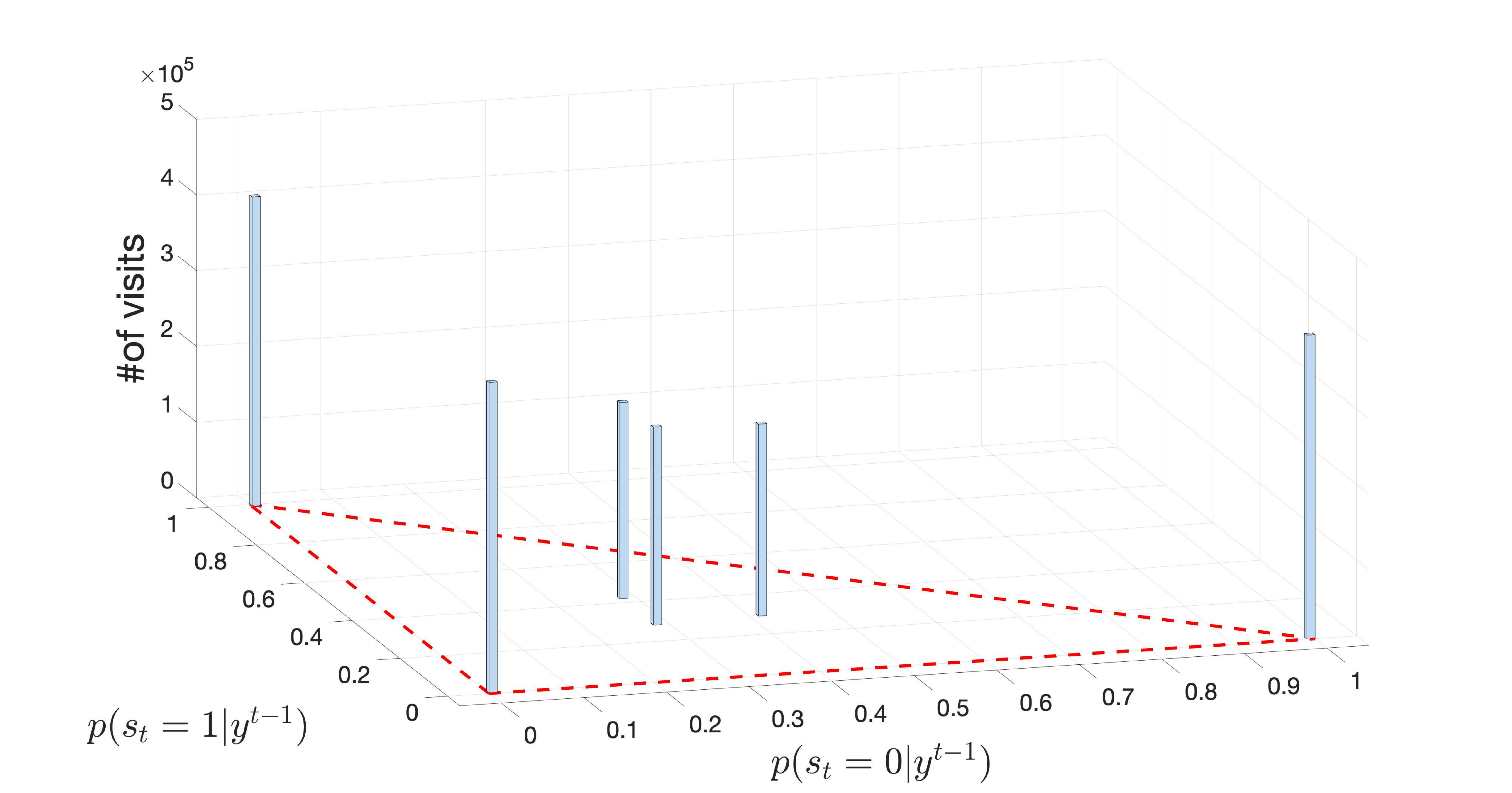}
    }
    \caption{State histogram of the policy as learnt by RL. The histogram is generated by a Monte-Carlo evaluation of the policy: (a) histogram of the policy after 1000 training iterations; (b) histogram of the policy after convergence.}
    \label{fig:state-hist}
\end{figure}

\subsubsection{Extracting the Structure of the Optimal Solution}
The output of the RL algorithm contains the actor, a parametric model of the input distribution of the channel. This network is used to obtain the structure of the solution by the following procedure.
First, it is used for a Monte-Carlo evaluation of length $n$ of the communication rate. During this evaluation, the MDP states and the channel outputs are recorded. These states are then clustered using common techniques, such as the k-means algorithm \cite{kmeans}. 
For instance, in Figure \ref{fig:state-hist} the MDP state histogram of the Ising channel with $|\cX|=3$  is shown, and it is clear that the estimated solution has only six discrete states.
Therefore, the sequence of MDP states $\{Z_i\}_{i=1}^n$ is converted into a sequence of auxiliary RVs $\{Q_i\}_{i=1}^n$ with a discrete alphabet $\cQ$, where each value in $\cQ$ forms a node of the Q-graph.
The transitions between nodes are determined uniquely\footnote{The disturbance is the only randomness of the transition between RL states.} by the channel outputs, and the corresponding test distribution $P_{Y_i|Q_{i-1}}=T_{Y|Q}$ is estimated by counting the channel output frequency at every Q-graph node. The Q-graph for $|\cX|=3$ is shown in Figure \ref{fig:q-graph}.
This completes the generation of the induced Q-graph and $T_{Y|Q}$. 

\begin{figure}[b]
    \centering
    \includegraphics[scale=1, trim=0 0 0 0, clip]{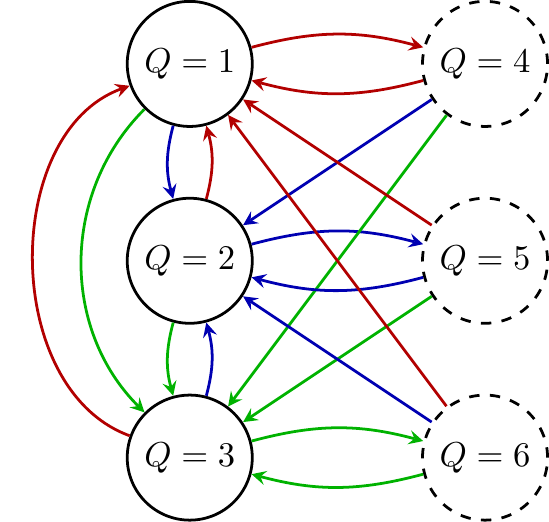}
    \caption{Q-graph showing the transitions between states as a function of the channel's output. Blue, red and green lines correspond to $Y=0,1,2$, respectively. States with solid lines and dashed lines indicate whether the channel state is known or unknown to the decoder, respectively.}
    \label{fig:q-graph}
\end{figure}

\subsubsection{Upper Bound Using the Extracted Structure}
The upper bound is derived by using the Q-graph and $T_{Y|Q}$ in the duality bound for the unifilar FSC with feedback, as presented in \cite{sabag2020duality}. 
The duality bound is given in the following theorem. 
\begin{theorem}\cite[Theorem 4]{sabag2020duality}\label{thm:duality-ub}
For any choice of  Q-graph and test distribution $T_{Y|Q}$, the feedback capacity of a strongly connected unifilar FSC is bounded by
\begin{equation}
    C_\mathsf{FB} \leq \lim_{n \to \infty} \max_{f(x^n\lVert y^n)} \max_{s_0,q_0}\frac{1}{n}\sum_{i=1}^n \bE\left[D_{KL}\left(P_{Y|X,S}\left(\cdot | x_i, S_{i-1}\right)\lVert T_{Y|Q}\left(\cdot|Q_{i-1}\right)\right)\right].
\end{equation}
The notation $f(x^n\lVert y^n) = \prod_i \mathds{1}\left[x_i = f_i \left(x^{i-1},y^{i-1}\right)\right]$ stands for the causal conditioning of deterministic functions.
\end{theorem}
\noindent The Q-graph transition function is denoted by $\phi: \cQ \times \cY \rightarrow \cQ$, where $\phi(q,y)$ is the node followed by a transition from node $Q=q$ when the channel output is $Y=y$.

\par The upper bound defines an infinite horizon average reward MDP, as described in Table \ref{tab:ub-mdp}, whose average reward is the upper bound on the feedback capacity in Theorem \ref{thm:duality-ub}.
\begin{table}[!h]
\caption{MDP Formulation of the Duality Upper Bound}
 \centering
 \begin{tabular}{|c | c|} 
 \hline
 state & $(s_{t-1}, q_{t-1})$  \\ 
 \hline
 action & $x_t$  \\
 \hline
 reward & $D_{KL}\left(P_{Y|X=x_t,S=s_{t-1}}\lVert T_{Y|Q=q_{t-1}}\right)$ \\
 \hline
 disturbance & $y_t$ \\
 \hline
\end{tabular}
\label{tab:ub-mdp}
\end{table}
Unlike the MDP of the feedback capacity of Theroem \ref{thm:capacity_unifilar}, this MDP has finite state and action spaces.
Therefore, its evaluation is tractable with DP algorithms, such as the value iteration algorithm. 
For this purpose, the corresponding Bellman equation is
\begin{align}\label{eqn:bellman_duality}
    \rho + V(s,q) = \max_x D_{KL}\left(P_{Y|X=x,S=s}\lVert T_{Y|Q=q}\right)+\sum_{y\in\cY}p(y|x,s)V\left(x,\phi(q,y)\right),
\end{align}
where the term $V(s,q)$ is the value function and $\rho$ is the average reward.

\par Since the state and action spaces are finite, the Bellman equation defines a finite set of non-linear equations. 
Removing the non-linearity is achieved by solving the Bellman equation numerically using the value iteration algorithm. 
The solution includes an estimate of the value function and the average reward, but more importantly, it provides a conjectured optimal policy 
\begin{equation*}
    x(s,q) = \argmax_x D_{KL}\left(P_{Y|X=x,S=s}\lVert T_{Y|Q=q}\right)+\sum_{y\in\cY}p(y|x,s)V^\ast\left(x,\phi(q,y)\right),
\end{equation*}
where $V^\ast$ is the estimated optimal value function. 
Substituting $x(s,q)$ in the Bellman equation converts it to a set of linear equations that are simple to solve and one obtains a conjecture of the optimal value function and average reward.
Finally, the conjectured value function and average reward are verified as the optimal (fixed point) solution using the Bellman equation to complete the bound. 

\par The bound is tested to be tight by verifying that the structure satisfies two conditions. The first is the Markov $Y^{i-1} - Q_{i-1} - Y_i$, which means that there exists an input distribution that visits only the MDP states that formed the Q-graph and yields an output distribution that satisfies $P_{Y_i|Y^{i-1}}=T_{Y|Q}$. The second condition is that the rate of this distribution equals the upper bound.
In that case, the bound is tight, which completes the proof.
In the next section, we demonstrate this methodology on the Ising channel with alphabet $|\cX| \leq 8$, and $|\cX| > 8$.

\subsection{Bounds on the Ising Channel}
\par In this section we present our results on the Ising channel. 
First, we derive the feedback capacity of the Ising channel with $|\cX| \leq 8$ by providing a tight upper bound. In this case, we also present a capacity achieving coding scheme.
Next, we provide an upper bound for $|\cX| > 8$. 
Finally, we provide an additional coding scheme and prove it is optimal for an asymptotic alphabet size.

\subsubsection{Capacity for $|\cX|\leq8$}\label{sec:ising-small}
\par After applying the RL algorithm, we obtain a model of the input distribution. We use this model to conduct a Monte-Carlo evaluation of the communication rate using the MDP formulation.
Then, the visited states are clustered using the k-means algorithm.
Each cluster is a distinct value of $\cQ$ that corresponds to an MDP state.
Let us denote each node in the graph by the tuple $(q_d, q_s)$, 
\begin{align}
    q_d &=  \begin{cases}
                1, & \text{decoder knows the channel state} \\
                0, & \text{decoder does not know the channel state}
                \end{cases} \\
    q_s &= \text{last known channel state},
\end{align}
where the decoder knows the channel state if $P_{S_t|Y^t}\left(\cdot|y^{t}\right)$ contains a symbol w.p. 1.
The Q-graph is defined by 
\begin{align}
\big(q_d^\prime,q_s^\prime\big) = 
\begin{cases}
                    (1,y) & q_d=0 \text{ or } q_d=1,y\neq q_s \\
                    (0,y) & q_d=1,y=q_s \\
\end{cases}.
\end{align}
Finally, $T_{Y|Q}$ is estimated by counting channel outputs at each node, and according to the transitions between nodes, edges are filled in the Q-graph. 
A parameterized version of $T_{Y|Q}$ is given in \eqref{eqn:tyq-lt8}.

\par The Q-graph and $T_{Y|Q}$ are plugged into the duality bound as described in the previous section. Then, the value iteration algorithm is applied on the upper bound to obtain $x(s,q)$. This allows the conversion of the Bellman equation into a set of linear equations. Consequently, we conjecture the value function and average reward, that are proven as optimal, as given in Lemma \ref{lemma:optimal_value_function_lt8}. 
\begin{lemma}\label{lemma:optimal_value_function_lt8}
For a fixed Q-graph and $T_{Y|Q}$ the function
\begin{align}
    V(s,q) &= \begin{cases}
                        \rho                     &, q_d=1, q_s = s  \\
                        1+1.5\rho                &, q_d=1, q_s \neq s \\
                        2\rho-1+\log(1+p)        &, q_d=0, q_s = s \\
                        1.5\rho +\log(1+p)       &, q_d=0, q_s\neq s
                     \end{cases}
\end{align}
and the constant $\rho = \frac{1}{2}\log \frac{1}{p}$ satisfy the Bellman equation. The variable $p$ is the only root of $x^4 - ((|\cX|-1)^4+4)x^3 +6x^2-4x +1 = 0$ that lies in $[0,1]$.
Equivalently, the optimal average reward can be rewritten as $$\rho = \max_{p \in [0,1]}  2\frac{H_2(p) + (1-p) \log\left(\lvert\cX\lvert-1\right)}{p+3}.$$
\end{lemma}
\noindent Lemma \ref{lemma:optimal_value_function_lt8} provides an upper bound on the feedback capacity, as given in Theorem \ref{thm:duality-ub}; its proof is given in Appendix \ref{sec:app-ub-small}.

\par The upper bound is verified to be tight by testing if $T_{Y|Q}$ is BCJR invariant. 
That is, there exists an input distribution whose corresponding output distribution satisfies $P_{Y_i|Y^{i-1}} = T_{Y|Q}$, with the same rate as the upper bound.
For this purpose, we conjecture the input distribution by averaging the actions at every Q-graph node.
This yields in the following input distribution:
\begin{equation}\label{eqn:opt_policy}
    p(x_t \lvert s_{t-1}, (q_d,q_s)) = 
        \begin{cases}
            p                       &, q_d=1, q_s = s_{t-1}, x_t=s_{t-1}  \\
            \frac{1-p}{|\cX|-1}     &, q_d=1, q_s = s_{t-1}, x_t\neq s_{t-1}  \\
            \text{arbitrary}        &, q_d=1, q_s \neq s_{t-1} \\
            1                       &, q_d=0, x_t=s_{t-1}  \\
            0                       &, q_d=0, x_t\neq s_{t-1}  
        \end{cases},
\end{equation}
which is a parameterized version of the numerical results. In words, when the decoder knows the channel state, the symbol repeats with probability $p$; otherwise, another symbol is chosen uniformly over all other symbols. When the decoder does not know the channel state, the state is transmitted again. 
Using the input distribution, the tightness of the bound is verified; this is stated in the following lemma, which completes the derivation of the feedback capacity and the proof of Theorem \ref{thm:ising_cfb}.
\begin{lemma}\label{lemma:ub-tight}
The input distribution in \eqref{eqn:opt_policy} is BCJR-invariant and its achievable rate is 
\begin{equation}
     I(X,S;Y|Q) = \max_{p \in [0,1]}  2\frac{H_2(p) + (1-p) \log\left(\lvert\cX\lvert-1\right)}{p+3}.
\end{equation}
Therefore, it serves as a tight lower bound as $C_\mathsf{FB}\geq \rho$.
\end{lemma}
The proof of this lemma is given in Appendix \ref{sec:app-ub-small}. The combination of Lemma \ref{lemma:optimal_value_function_lt8} and Lemma \ref{lemma:ub-tight} concludes the proof of Theorem \ref{thm:ising_cfb}.
\vspace{0.5cm}

\subsubsection{Upper Bound for $|\cX| > 8$}\label{sec:ising-large}
\par The Q-graph that was optimal for $|\cX| \leq 8$ is not optimal for $|\cX|>8$ since the upper bound in this region is not tight. Therefore, we conducted an RL simulation on the Ising channel with an alphabet $|\cX| = 9$ to obtain a new Q-graph. In this case, the structure of the solution has a complex histogram, and hence the structure cannot be fully recovered. Instead, we extract a subset of states where most of the transitions occur. This results in a graph with 12 nodes.
\begin{figure}[b]
    \centering
    \includegraphics[scale=1, trim=0 0 0 0, clip]{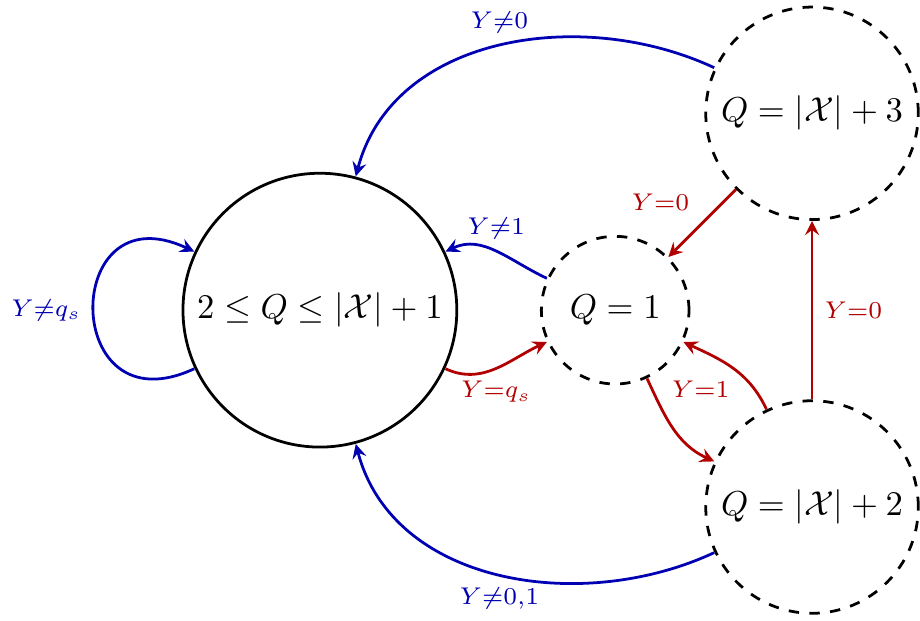}
    \caption{Q-graph of the Ising channel with $\lvert \cX\lvert=9$. The node $Q=1$ represents a state where the decoder knows the channel state. At other nodes the channel state is not known exactly to the decoder.
    }
    \label{fig:q-graph9}
\end{figure}

\par We generalize the structure for a general alphabet $|\cX|>2$, as depicted in Figure \ref{fig:q-graph9}. For simplicity of the graph, all the nodes where the decoder knows the state of the channel are merged into node $2\leq Q \leq |\cX|+1$.
Next, we define the Q-graph with $|\cX|+3$ nodes. Every node where $Q\neq 1$ has a channel state with which it is associated. Therefore, all nodes except $Q=1$ are denoted by a tuple $(q, q_s)$; for $2\leq q \leq |\cX|+1$, $q_s=q-2$, for $q=|\cX|+2,|\cX|+3$ we set $q_s=0,1$ (arbitrary choice), respectively. Transitions between nodes as a function of the channel output are depicted in Figure \ref{fig:q-graph9}.
The test distribution $T_{Y|Q}$ is estimated subsequently and its parameterized version is given in \eqref{eqn:tyq-gt8}.
\par Using the same methodology as in Section \ref{sec:analytic-method}, we use the Q-graph and $T_{Y|Q}$ in the upper bound to obtain an upper bound whose Bellman solution is presented in the following lemma.
\begin{lemma}\label{lemma:optimal_value_function_g8}
For a fixed Q-graph and $T_{Y|Q}$ the function 
\begin{align}
    V(s,q) &= \begin{cases}
                        \rho                        &, q=1  \\
                        2\rho -\log|\cX|            &, q\neq 1, q_s=s  \\
                        2\rho + 2 -\log|\cX|        &, q\neq 1, q_s\neq s  \\
                     \end{cases}
\end{align}
and the constant $\rho = \frac{1}{2}\log \frac{|\cX|}{p}$ satisfy the Bellman equation. The variable $p$ is the root of $x^2-(2 + \frac{(|\cX|-1)^2}{16|\cX|})x +1 = 0$ that lies in $[0,1]$.
\end{lemma}
In Appendix \ref{sec:app-ub-large} the proof of Lemma \ref{lemma:optimal_value_function_g8} is given, which completes the proof of the upper bound for $|\cX|>8$ as given in Theorem \ref{thm:ub-large_cardinality}. 
\vspace{0.5cm}

\subsubsection{Coding Scheme for $|\cX|\leq 8$}\label{sec:ising-code-small}
The insights from the numerical results led us to derive a capacity achieving coding scheme for $|\cX|\leq 8$, as stated in Theorem \ref{thm:code-scheme}. This code is a generalization of the optimal coding scheme for $|\mathcal X| =2$  that was presented in \cite{Ising_channel}. 
\begin{proof}[Proof of Theorem \ref{thm:code-scheme}]
The achievable rate is computed by dividing the entropy rate of input symbols by the expected channel uses per one symbol.
The entropy rate is computed by the source statistics,
\begin{align} \label{eqn:input-len}
    H(\nu_i|\nu_{i-1}) = H_2(p) + (1-p) \log\left(\lvert\cX\lvert-1\right).
\end{align} 
The expected channel uses per one symbol $\nu_i$ is
\begin{align}\label{eqn:code-len}
    \bE \left[ L\right] = p \cdot 2 + (1-p) \cdot 1.5.
\end{align}
That is since when $\nu_i = \nu_{i-1}$, the symbol is sent twice, and when $\nu_i \neq \nu_{i-1}$, the symbol is sent once or twice with equal probability.
The proof is completed by dividing \eqref{eqn:input-len} by \eqref{eqn:code-len} and taking a maximum over $p$. 
\end{proof}
\vspace{0.5cm}

\subsubsection{Asymptotic Coding Scheme}\label{sec:ising-asymp}
We present the proof of Theorem \ref{thm:ub-general}. First, we show the upper bound $C_\mathsf{FB}(\cX) \leq \frac{3}{4}\log(|\cX|)$. Then, we show a simple coding scheme with rate $R(\cX)=\frac{3}{4}\log\frac{\lvert\cX\lvert}{2}$ to complete the proof. 
\begin{proof}[Proof of upper bound in Theorem \ref{thm:ub-general}]\label{sec:ising-asymp-ub}
Let $W^n$ be a sequence of RVs that is defined by
\begin{equation}
    W_i =   \begin{cases}
                1   &, Y_i = X_i   \\
                0   &, Y_i = X_{i-1}
            \end{cases},
\end{equation}
equivalently, $W_i$ indicates whether the output of the channel is the current input or the previous one. By the channel definition $W^n$ is an i.i.d. sequence of RVs, independent of the message $M$, where $W_i \sim Ber(0.5)$.
Now, consider a series of achievable codes $(n,2^{nR})$ with rate $R$, an encoder $X_i = f_i(M,Y^{i-1})$ and a decoder $\hat{M}_i = g_i(Y^i)$ with $\mathbb{P}(M \neq \hat{M}_n)\overset{n\rightarrow\infty}{\longrightarrow} 0 $. A converse for the feedback capacity is then obtained by the following steps:
\begin{align}
    nR  &= H(M) \\
        &\overset{(a)}{=} H(M|W^n) \\
        &= H(M|W^n) - H(M|Y^n,W^n) + H(M|Y^n,W^n) \\
        &= I(M;Y^n|W^n) + \epsilon_n\\
        &\overset{(b)}{\leq} H(Y^n|W^n) + \epsilon_n\\
        &= \sum_{i=1}^n H(Y_i|Y^{i-1},W^n) + \epsilon_n\\
        &\overset{(c)}{\leq} \sum_{i=1}^n H(Y_i|Y_{i-1},W_i, W_{i-1}) + \epsilon_n\\
        &\overset{(d)}{\leq} n \frac{3}{4} \log |\cX| + \epsilon_n \label{eqn:general-ub-proof}
\end{align}
where (a) follows from the independence of $M,W^n$, (b) follows from the non-negativity of entropy, (c) follows from the fact that conditioning reduces entropy, and (d) is due to 
\begin{align}
    H(Y_i|Y_{i-1},W_i, W_{i-1}) &= \sum_{w_0,w_1 \in \{0,1\}^2}p(w_0,w_1) H(Y_i|Y_{i-1},W_i=w_1, W_{i-1}=w_0) \\
        &= \frac{1}{4}\sum_{w_0,w_1 \in \{0,1\}^2} H(Y_i|Y_{i-1},W_i=w_1, W_{i-1}=w_0).
\end{align}
By the channel definition, when $W_i=0, W_{i-1}=1$ it follows that $Y_i = Y_{i-1}$ and therefore $H(Y_i|Y_{i-1},W_i=0, W_{i-1}=1)=0 $. For $w_0,w_1 \in \{0,1\}^2 \setminus \{0,1\}$, we bound $H(Y_i|Y_{i-1},W_i=w_1, W_{i-1}=w_0) \leq \log |\cY| = \log |\cX|$.
Combining the results we obtain $\mathsf{C}_{fb} \leq \frac{3}{4} \log |\cX| + \frac{\epsilon_n}{n}$. 
According to Fano's inequality, $\lim_{n\rightarrow\infty} \frac{\epsilon_n}{n} = 0$. Thus, by
taking the limit we derive the desired upper bound.
\end{proof}

We show next a simple coding scheme that is asymptotically better than the capacity achieving coding scheme from Theorem \ref{thm:code-scheme}. The following proof describes the coding scheme and computes its rate.

\begin{algorithm}[H]
    \caption{Coding scheme for asymptotic $|\cX|$}
    \label{alg:code-scheme-asymp}
    \textbf{Code construction and initialization:}
    \begin{itemize}
    \item[-] Partition $\cX$ into two equal-sized disjoint sets $\cX_a, \cX_b$ (up to one symbol)
    \item[-] Transform a message of $n$ bits into a stream of symbols from $\mathcal X$, denoted by $\nu_1\nu_2\dots$ with the following statistics:
    \begin{equation}
        \nu_i = \begin{cases}
                    \text{Unif}[\cX_b] &, i \text{ even}\\
                    \text{Unif}[\cX_a] &, i \text{ odd}
                \end{cases}
    \end{equation}
    In words, symbols are drawn uniformly and interchangeably from $\cX_a,\cX_b$. Denote the source buffers for symbols from $\cX_a,\cX_b$ as $X_a, X_b$, respectively.
    \item[-] Generate two output buffers and the decoder $Y_a,Y_b$
    \item[-] Transmit a symbol twice to set the initial state of the channel $s_0$
    \end{itemize}
    \algrule
    \textbf{Encoder:}
    \begin{algorithmic}
    \STATE Transmit $\nu_t$ and observe $y_t$
    \IF{$y_t = \nu_t$ and $y_{t-1} = \nu_{t-2}$}
        \STATE return $\nu_{t-1}$ to the top of its source buffer
    \ENDIF
    \end{algorithmic}
    \algrule
    \textbf{Decoder:}
    \begin{algorithmic}
    \STATE Receive $y_t$
    \IF{$y_t \neq y_{t-1}$}
    \IF{$y_t \in \cX_a$}
        \STATE store $y_t$ in $Y_a$ 
    \ELSE
        \STATE store $y_t$ in $Y_b$
    \ENDIF
    \ENDIF
    \end{algorithmic}
\end{algorithm}

\begin{proof}[Proof of coding scheme in Theorem \ref{thm:ub-general}]\label{sec:ising-asymp-code}
The coding scheme partitions the alphabet into two distinct sets that are assigned uniquely as the sources for odd and even transmission steps.
Thus, the encoder sends interchangeably from both sets, which enables the decoder to distinguish whether the channel output was the input or the state of the channel. 

\underline{Code analysis:}
The rate, $\mathsf{R}\left(\cX\right)$ of the code is computed by dividing the entropy rate of the source by the expected channel uses per one symbol. The entropy rate is $H(\cX) = \log\frac{\lvert\cX\lvert}{2}$. 
Let $L$ denote the number of channel uses per one symbol. We compute $\bE[L]$ by conditioning on $W_0, W_1$, RVs that indicate if the output of the channel is the input or the state, as defined in Section \ref{sec:ising-asymp-ub}. Consequently, it follows that
\begin{align}
    \bE \left[L\right]  &= \sum_{w_0,w_1 \in \{0,1\}^2} \bP[W_0=w_0,W_1=w_1] \bE \left[L\lvert W_0=w_0,W_1=w_1 \right] \\
                        &= \underbrace{0.5\cdot 1}_{w_0=1} + \underbrace{0.25\cdot1}_{w_0=0,w_1=0}+\underbrace{0.25\bE \left[L \lvert W_0=0,W_1=1\right]}_\text{symbol is transmitted again} \\
                        &= 0.75 + 0.25(1+ \bE \left[L\right]).
\end{align}
By rearranging we obtain $\bE\left[L\right] = \frac{4}{3}$. Finally, by dividing the entropy rate by the expected channel uses per one symbol, the rate is obtained as $\mathsf{R}\left(\cX\right)= \frac{3}{4}\log\frac{\lvert\cX\lvert}{2}$.
\end{proof}

\section{Discussion and Conclusions}\label{sec:conclusions}
\par We proposed a new methodology to compute the feedback capacity of unifilar FSCs. RL is proposed instead of classic DP algorithms due its ability to evaluate MDPs with continuous state and action spaces. Two RL algorithms are proposed to evaluate numerically the feedback capacity. The numerical results form the basis for conjecturing the structure of the optimal solution via a Q-graph and a corresponding $T_{Y|Q}$. The structure is used in the duality bound to obtain an analytic expression of the upper bound, which is tested as tight by verifying that $T_{Y|Q}$ is BCJR invariant.

\par We applied this methodology to obtain analytic results over the Ising channel with a general alphabet. For $ |\cX| \leq 8$, we found the analytic solution of the feedback capacity and derived a capacity achieving coding scheme. For $|\cX|>8$, the structure in the numerical results enabled us to obtain an upper bound, but we did not manage to verify whether it is tight or not.

\par An interesting observation is the change of the structure of the solution as the alphabet size increases. Mathematically, the capacity achieving coding scheme for $|\cX| \leq 8$ is optimal for $\cX : \mathsf{R}\left(\cX\right)\leq 2$, as shown in Appendix \ref{sec:app-ub-small}. This implies that the transmission policy might differ for the same channel with increasing alphabet size. We visualize this observation in Figure \ref{fig:res-summary}, where the efficiency of the coding scheme in Theorem \ref{thm:code-scheme}, the numerical lower bound from the RL simulation, and the analytic upper bound in Theorem \ref{thm:ub-large_cardinality} are compared. It is visible that the solution for $ |\cX| \leq 8$ saturates at $\frac{2}{3}$, where the numeric lower bound keeps improving when the alphabet increases.
This phenomenon is observed when increasing the alphabet size, which emphasizes the importance of developing effective tools for channels with large alphabet sizes.

\section*{Acknowledgements}
This work was supported by the German Research Foundation (DFG) via
the German-Israeli Project Cooperation [DIP] and by the ISF research
grant 818/17. The work of O. Sabag is partially supported by the ISEF international postdoctoral fellowship.

\bibliography{ref}

\begin{thebibliography}{10}
\providecommand{\url}[1]{#1}
\csname url@samestyle\endcsname
\providecommand{\newblock}{\relax}
\providecommand{\bibinfo}[2]{#2}
\providecommand{\BIBentrySTDinterwordspacing}{\spaceskip=0pt\relax}
\providecommand{\BIBentryALTinterwordstretchfactor}{4}
\providecommand{\BIBentryALTinterwordspacing}{\spaceskip=\fontdimen2\font plus
\BIBentryALTinterwordstretchfactor\fontdimen3\font minus
  \fontdimen4\font\relax}
\providecommand{\BIBforeignlanguage}[2]{{%
\expandafter\ifx\csname l@#1\endcsname\relax
\typeout{** WARNING: IEEEtran.bst: No hyphenation pattern has been}%
\typeout{** loaded for the language `#1'. Using the pattern for}%
\typeout{** the default language instead.}%
\else
\language=\csname l@#1\endcsname
\fi
#2}}
\providecommand{\BIBdecl}{\relax}
\BIBdecl

\bibitem{aharoni2019computing}
Z.~Aharoni, O.~Sabag, and H.~H. Permuter, ``Computing the feedback capacity of
  finite state channels using reinforcement learning,'' in \emph{2019 IEEE
  International Symposium on Information Theory (ISIT)}.\hskip 1em plus 0.5em
  minus 0.4em\relax IEEE, 2019, pp. 837--841.

\bibitem{Q-UB}
O.~Sabag, H.~H. Permuter, and H.~D. Pfister, ``A single-letter upper bound on
  the feedback capacity of unifilar finite-state channels,'' \emph{IEEE Trans.
  Inf. Theory}, vol.~63, no.~3, pp. 1392--1409, March 2017.

\bibitem{sabag2020duality}
O.~Sabag and H.~H. Permuter, ``The duality upper bound for unifilar
  finite-state channels with feedback,'' in \emph{International Zurich Seminar
  on Information and Communication (IZS 2020). Proceedings}.\hskip 1em plus
  0.5em minus 0.4em\relax ETH Zurich, 2020, pp. 68--72.

\bibitem{PermuterCuffVanRoyWeissman08}
H.~H. Permuter, P.~Cuff, B.~V. Roy, and T.~Weissman, ``Capacity of the trapdoor
  channel with feedback,'' \emph{IEEE Trans. Inf. Theory}, vol.~54, no.~7, pp.
  3150--3165, Jul. 2009.

\bibitem{Berger90IsingChannel}
T.~Berger and F.~Bonomi, ``Capacity and zero-error capacity of {I}sing
  channels,'' \emph{IEEE Trans. Inf. Theory}, vol.~36, pp. 173--180, 1990.

\bibitem{ising1925beitrag}
E.~Ising, ``Beitrag zur theorie des ferromagnetismus,'' \emph{Zeitschrift
  f{\"u}r Physik}, vol.~31, no.~1, pp. 253--258, 1925.

\bibitem{Ising_artyom_IT}
A.~Sharov and R.~M. Roth, ``On the capacity of generalized ising channels,''
  \emph{IEEE Trans. on Inf. Theory}, Dec. 2016.

\bibitem{Ising_channel}
O.~Elishco and H.~Permuter, ``Capacity and coding for the {I}sing channel with
  feedback,'' \emph{IEEE Trans. Inf. Theory}, vol.~60, no.~9, pp. 5138--5149,
  Sep. 2014.

\bibitem{ddpg}
T.~P. Lillicrap, J.~J. Hunt, A.~Pritzel, N.~Heess, T.~Erez, Y.~Tassa,
  D.~Silver, and D.~Wierstra, ``Continuous control with deep reinforcement
  learning,'' \emph{arXiv preprint arXiv:1509.02971}, 2015.

\bibitem{1054929}
T.~Cover, ``Enumerative source encoding,'' \emph{IEEE Trans. Inf. Theory},
  vol.~19, no.~1, pp. 73--77, January 1973.

\bibitem{sutton2018reinforcement}
R.~S. Sutton and A.~G. Barto, \emph{Reinforcement learning: An
  introduction}.\hskip 1em plus 0.5em minus 0.4em\relax MIT press, 2018.

\bibitem{value_iteration_bellman}
R.~Bellman, ``A markovian decision process,'' \emph{Indiana Univ. Math. J.},
  vol.~6, pp. 679--684, 1957.

\bibitem{importance_sampling}
W.~K. Hasting, ``{Monte Carlo sampling methods using Markov chains and their
  applications},'' \emph{Biometrika}, vol.~57, no.~1, pp. 97--109, 04 1970.

\bibitem{dropout}
\BIBentryALTinterwordspacing
N.~Srivastava, G.~Hinton, A.~Krizhevsky, I.~Sutskever, and R.~Salakhutdinov,
  ``Dropout: A simple way to prevent neural networks from overfitting,''
  \emph{Journal of Machine Learning Research}, vol.~15, no.~56, pp. 1929--1958,
  2014. [Online]. Available:
  \url{http://jmlr.org/papers/v15/srivastava14a.html}
\BIBentrySTDinterwordspacing

\bibitem{kmeans}
S.~Lloyd, ``Least squares quantization in {PCM},'' \emph{IEEE Trans. Inf.
  Theory}, vol.~28, no.~2, pp. 129--137, 1982.

\end{thebibliography}
\bibliographystyle{IEEEtran}
\newpage

\begin{appendices}
\section{Upper Bound on the Feedback Capacity for $\lvert\cX\lvert \leq 8$}\label{sec:app-ub-small}
\begin{proof}[Proof of Lemma \ref{lemma:optimal_value_function_lt8}]
The conditional distribution $T \lb y \lvert q \rb$ that corresponds to the Q-graph extracted from the numerical results is given by 
\begin{equation} \label{eqn:tyq-lt8}
    T\lb y\lvert q \rb =\begin{cases}
    \frac{1+p}{2}                             &, q_d=1, q_s = y  \\
    \frac{1-p}{2\lb \cardin-1\rb}             &, q_d=1, q_s \neq y \\
    \frac{2p}{1+p}                       &, q_d=0, q_s=y  \\
\frac{1-p}{\lb \cardin-1 \rb \lb 1+p\rb} &, q_d=0, q_s\neq y
                        \end{cases} .
\end{equation}
Next, we verify that the Bellman equation holds for every $(s,q)$
\begin{align}
    \rho + V(s,q) = \max_x D_{KL}\left(P_{Y|X=x,S=s}\lVert T_{Y|Q=q}\right)+\sum_{y\in\cY}p(y|x,s)V\left(x,\phi(q,y)\right).
\end{align}
\noindent We will start by computing the Bellman equation operator for $q_d=1, q_s=s$. For $x=s$,
\begin{align}
    V(s,q)  &= -\rho+\sum_{y=x,s} p(y|x,s)\left[\log\left(\frac{p(y|x,s)}{T(y|q)}\right) + V(x,\phi(q,y))\right] \\
            &= -\rho+1 \left[\log\left(\frac{1}{\frac{1+p}{2}}\right) + V\left(s,(1,s)\right)\right] \\
            &= -\rho+1-\log(1+p) + 2\rho-1+\log(1+p) \\
            &= \rho.
\end{align}
Further, if $x\neq s$
\begin{align}
    V(s,q)  &= -\rho+\sum_{y=x,s} p(y|x,s)\left[\log\left(\frac{p(y|x,s)}{T(y|q)}\right) + V(x,\phi(q,y))\right] \\
            &= -\rho+\frac{1}{2} \underbrace{\left[\log\left(\frac{\frac{1}{2}}{\frac{1+p}{2}}\right) + V(x,(0,s))\right]}_{y=s} +  \frac{1}{2} \underbrace{\left[\log\left(\frac{\frac{1}{2}}{\frac{1-p}{2\lb \cardin-1\rb}}\right) + V(x,(0,x))\right]}_{y=x}\\
            &= -\rho+\frac{1}{2} \left[ -\log(1+p) + 1.5\rho +\log(1+p) -\log\left(\frac{1-p}{\lb \cardin-1\rb}\right)+\rho\right] \\
            &= -\rho+\frac{1}{2} \left[ -\log(1+p) + 1.5\rho +\log(1+p) +1.5\rho+\rho\right] \\ &= \rho.
\end{align}
Since, for all $x$ the operator is equal, the Bellman equation is satisfied.\\
\vspace{10pt}
\noindent For $q_d=1, q_s \neq s$
\newline
$x=s$
\begin{align}
    V(s,q)  &= -\rho+\sum_{y=x,s} p(y|x,s)\left[\log\left(\frac{p(y|x,s)}{T(y|q)}\right) + V(x,\phi(q,y))\right] \\
            &= -\rho+1 \left[\log\left(\frac{1}{\frac{1-p}{2\lb \cardin-1\rb}}\right) + V\left(s,(1,s)\right)\right] \\
            &= -\rho+1-\log\left(\frac{1-p}{\lb \cardin-1\rb}\right)+\rho = 1+1.5\rho.
\end{align}
$x\neq s$, $q_s = x$
\begin{align}
    V(s,q)  &= -\rho+\sum_{y=x,s} p(y|x,s)\left[\log\left(\frac{p(y|x,s)}{T(y|q)}\right) + V(x,\phi(q,y))\right] \\
            &= -\rho+\frac{1}{2} \underbrace{\left[\log\left(\frac{\frac{1}{2}}{\frac{1-p}{2\lb \cardin-1\rb}}\right) + V(x,(1,s))\right]}_{y=s} +  \underbrace{\frac{1}{2} \left[\log\left(\frac{\frac{1}{2}}{\frac{1+p}{2}}\right) + V(x,(0,x))\right]}_{y=x}\\
            &= -\rho+\frac{1}{2} \left[-\log\left(\frac{1-p}{\lb \cardin-1\rb}\right) + 1+1.5\rho)\right] +  \frac{1}{2} \left[-\log\left(1+p\right) + 2\rho-1+\log(1+p)\right] \\
            &= -\rho+\frac{1}{2} \left[1.5\rho + 1+1.5\rho)\right] +  \frac{1}{2} \left[2\rho-1\right] = 1.5\rho.
\end{align}
$x\neq s$, $q_s \neq x$
\begin{align}
    V(s,q)  &= -\rho+\sum_{y=x,s} p(y|x,s)\left[\log\left(\frac{p(y|x,s)}{T(y|q)}\right) + V(x,\phi(q,y))\right] \\
            &= -\rho+\frac{1}{2} \underbrace{\left[\log\left(\frac{\frac{1}{2}}{\frac{1-p}{2\lb \cardin-1\rb}}\right) + V(x,(1,s))\right]}_{y=s} +  \underbrace{\frac{1}{2} \left[\log\left(\frac{\frac{1}{2}}{\frac{1-p}{2\lb \cardin-1\rb}}\right) + V(x,(0,x))\right]}_{y=x}\\
            &= -\rho+\frac{1}{2} \left[-\log\left(\frac{1-p}{\lb \cardin-1\rb}\right) + 1+1.5\rho)\right] +  \frac{1}{2} \left[-\log\left(\frac{1-p}{\lb \cardin-1\rb}\right) + \rho\right] \\
            &= -\rho+\frac{1}{2} \left[1.5\rho + 1+1.5\rho)\right] +  \frac{1}{2} \left[1.5\rho + \rho\right] = 1.75\rho + 0.5. 
\end{align}
Hence, for any cardinality that satisfies $\rho < 2$ the Bellman equation holds ($1.5\rho+1>1.75\rho+0.5$).\\

\vspace{10pt}
\noindent $q_d=0, q_s = s$
\newline
$x=s$
\begin{align}
    V(s,q)  &= -\rho+\sum_{y=x,s} p(y|x,s)\left[\log\left(\frac{p(y|x,s)}{T(y|q)}\right) + V(x,\phi(q,y))\right] \\
            &= -\rho+1 \left[\log\left(\frac{1}{\frac{2p}{1+p}}\right) + V\left(s,(1,s)\right)\right] \\
            &= -\rho+\log(1+p)-1-\log(p) + \rho = 2\rho-1+\log(1+p).
\end{align}
$x\neq s$
\begin{align}
    V(s,q)  &= -\rho+\sum_{y=x,s} p(y|x,s)\left[\log\left(\frac{p(y|x,s)}{T(y|q)}\right) + V(x,\phi(q,y))\right] \\
            &= -\rho+\frac{1}{2} \underbrace{\left[\log\left(\frac{\frac{1}{2}}{\frac{2p}{1+p}}\right) + V(x,(1,s))\right]}_{y=s} +  \frac{1}{2} \underbrace{\left[\log\left(\frac{\frac{1}{2}}{\frac{1-p}{\lb \cardin-1 \rb \lb 1+p\rb}}\right) + V(x,(1,x))\right]}_{y=x}\\
            &= -\rho+\frac{1}{2} \left[\log(1+p) - 2 - \log(p)+1+1.5\rho\right] +  \frac{1}{2}\left[\log(1+p)-1-\log\left(\frac{1-p}{\lb \cardin-1 \rb}\right)+\rho\right]\\
            &= -\rho+\log(1+p)+\frac{1}{2} \left[ - 2 +2\rho+1+1.5\rho\right] +  \frac{1}{2}\left[-1+1.5\rho+\rho\right] = 2\rho-1+\log(1+p).
\end{align}
In that case the Bellman equation is satisfied.\\

\vspace{10pt}
\noindent $q_d=0,q_s \neq s$
\newline
$x=s$
\begin{align}
    V(s,q)  &= -\rho+\sum_{y=x,s} p(y|x,s)\left[\log\left(\frac{p(y|x,s)}{T(y|q)}\right) + V(x,\phi(q,y))\right] \\
            &= -\rho+1 \left[\log\left(\frac{1}{\frac{1-p}{\lb \cardin-1 \rb \lb 1+p\rb}}\right) + V\left(s,(1,s)\right)\right] \\
            &= -\rho+\log(1+p)-\log\left(\frac{1-p}{\lb \cardin-1 \rb}\right) + \rho = 1.5\rho+\log(1+p).
\end{align}
$x\neq s, q_s = x$
\begin{align}
        V(s,q)  &= -\rho+\sum_{y=x,s} p(y|x,s)\left[\log\left(\frac{p(y|x,s)}{T(y|q)}\right) + V(x,\phi(q,y))\right] \\
            &= -\rho+\frac{1}{2} \underbrace{\left[\log\left(\frac{\frac{1}{2}}{\frac{1-p}{\lb \cardin-1 \rb \lb 1+p\rb}}\right) + V(x,(1,s))\right]}_{y=s} +  \frac{1}{2} \underbrace{\left[\log\left(\frac{\frac{1}{2}}{\frac{2p}{1+p}}\right) + V(x,(1,x))\right]}_{y=x}\\
            &= -\rho+\frac{1}{2} \left[\log(1+p)-1-\log\left(\frac{1-p}{\lb \cardin-1 \rb}\right)+1+1.5\rho\right] +  \frac{1}{2}\left[\log(1+p)-2-\log(p)+\rho\right]\\
            &= -\rho+\log(1+p)+\frac{1}{2} \left[-1+1.5\rho+1+1.5\rho\right] +  \frac{1}{2}\left[-2+2\rho+\rho\right] = 2\rho-1+\log(1+p).
\end{align}

$x\neq s$, $q_s \neq x$
\begin{align}
    V(s,q)  &= -\rho+\sum_{y=x,s} p(y|x,s)\left[\log\left(\frac{p(y|x,s)}{T(y|q)}\right) + V(x,\phi(q,y))\right] \\
        &= -\rho+\frac{1}{2} \underbrace{\left[\log\left(\frac{\frac{1}{2}}{\frac{1-p}{\lb \cardin-1 \rb \lb 1+p\rb}}\right) + V(x,(1,s))\right]}_{y=s} +  \frac{1}{2} \underbrace{\left[\log\left(\frac{\frac{1}{2}}{\frac{1-p}{\lb \cardin-1 \rb \lb 1+p\rb}}\right) + V(x,(1,x))\right]}_{y=x}\\
        &= -\rho+\frac{1}{2} \left[\log(1+p)-1-\log\left(\frac{1-p}{\lb \cardin-1 \rb}\right)+1+1.5\rho\right] +  \frac{1}{2}\left[\log(1+p)-1-\log\left(\frac{1-p}{\lb \cardin-1 \rb}\right)+\rho\right]\\
        &= -\rho+\log(1+p)-1+1.5\rho+0.5+\frac{5}{4}\rho = \frac{7}{4}\rho-0.5+\log(1+p).
\end{align}
Here, too, the Bellman equation is satisfied for $\rho<2$.

\end{proof}

\begin{proof}[Proof of Lemma \ref{lemma:ub-tight}]
We show that the Markov $Y^{i-1}-Q_{i-1}-Y_i$ holds, and since $Q_i = \Phi(Y^{i-1})$ it is enough to show that $T_{Y|Q}=P_{Y_i|Y^{i-1}}$. 
Using the MDP state for every node of the Q-graph we obtain the following relation:
\begin{align}
    z_{i-1} = \begin{cases}
                [0,\dots,1,\dots,0] &, q_d=1 \\
                [\frac{1-p}{(1+p)(|\cX|-1)},\dots,\frac{2p}{1+p},\dots,\frac{1-p}{(1+p)(|\cX|-1)}] &, q_d=0
    \end{cases},
\end{align}
where the unique index corresponds to index $q_s$. Now we plug the input distribution into \eqref{eqn:opt_policy} to get the following representation:
\begin{align}
    p(y_i|y^{i-1}) = \sum_{x_i,s_{i-1}} z_{i-1}u_i p(y_i|x_i,s_{i-1}),
\end{align}
to verify the desired relation.

\noindent For $q_d=1, q_s=y$
\begin{align}
    \sum_{x,s} z_{i-1}(s)u_i(x,s)p(y|x,s) &= \sum_{x} u_i(x,y)p(y|x,y) \\
    &= \underbrace{p}_{x=y} + \underbrace{(|\cX|-1)\frac{1-p}{|\cX|-1} \frac{1}{2}}_{x\neq y} = \frac{1+p}{2}.
\end{align}

\noindent For $q_d=1, q_s\neq y$
\begin{align}
    \sum_{x,s} z_{i-1}(s)u_i(x,s)p(y|x,s) &= \sum_{x} u_i(x,y)p(y|x,y) \\
    &= \underbrace{\frac{1-p}{|\cX|-1} \frac{1}{2}}_{x=y}  .
\end{align}

\noindent For $q_d=0, q_s = y$
\begin{align}
    \sum_{x,s} z_{i-1}(s)u_i(x,s)p(y|x,s) &= \frac{2p}{1+p}\underbrace{\sum_{x} u_i(x,y)p(y|x,y)}_{=1}+ \frac{1-p}{(1+p)(|\cX|-1)}\underbrace{\sum_{x,s\neq y} u_i(x,s)p(y|x,s)}_{=0} \\
    &= \frac{2p}{1+p}.
\end{align}

\noindent For $q_d=0, q_s \neq y$
\begin{align}
    \sum_{x,s} z_{i-1}(s)u_i(x,s)p(y|x,s) &= \frac{2p}{1+p}\underbrace{\sum_{x} u_i(x,y)p(y|x,y)}_{=0}+ \frac{1-p}{(1+p)(|\cX|-1)}\underbrace{\sum_{x,s\neq y} u_i(x,s)p(y|x,s)}_{=1} \\
    &= \frac{1-p}{(1+p)(|\cX|-1)}.
\end{align}
Since the Markov $Y^{i-1}-Q_{i-1}-Y_i$ holds, the feedback capacity is converted into a single-letter expression $I(X,S;Y|Q)$ as shown in \cite{Q-UB}. First, we use $p(x|s,q),p(y|x,s)\mathds{1}(s^\prime = x) \mathds{1}(q^\prime=\phi(q,s))$ to compute the transition matrix of the Markov $S_0,Q_0-S_1,Q_1-\dots$ and compute its stationary distribution. It is given by
\begin{align}
    \pi(s,q) &= \begin{cases}
                        \frac{2}{\cardin \lb p+3\rb}                        &, q_d = 1, q_s = s \\
                        0                                                   &, q_d = 1, q_s \neq s \\
                        \frac{2p}{\cardin \lb p+3\rb}                       &, q_d = 0, q_s = s \\
                        \frac{1-p}{\cardin \lb \cardin-1 \rb \lb p+3\rb}    &, q_d = 0, q_s \neq s \\
                     \end{cases}.
\end{align}
Next, we compute the rate 
\begin{align*}
    H(Y|Q)  &= \sum_{q=(q_d,q_s)} \pi(q)H(Y|Q=q) \\
            &= \sum_{q=(1,q_s)} \pi(q)H(Y|Q=q) + \sum_{q=(0,q_s)} \pi(q)H(Y|Q=q)\\
            &= \frac{2}{p+3} H(Y|Q=(1,q_s)) + \frac{p+1}{p+3} H(Y|Q=(0,q_s)) \\
            &= \frac{2}{p+3}\lb H_2(p)+(1-p)\log(\cardin-1)\rb + 2\frac{1-p}{p+3}, \\
    H(Y|X,S)    &= -\sum_{x,s,y,q} \pi(s,q)p(x|s,q)p(y|x,s)\log p(y|x,s) \\
                &= -\cardin\sum_{q=(1,q_s),x,s,y} \pi(s,q)p(x|s,q)p(y|x,s)\log p(y|x,s) + \\
                &\quad -\cardin\sum_{q=(0,q_s),x,s,y} \pi(s,q)p(x|s,q)p(y|x,s)\log p(y|x,s) \\
                &= -\frac{2}{p+3} \sum_{q=(1,q_s),x,y} p(x|q_s,q)p(y|x,q_s)\log p(y|x,q_s) - \\
                &\quad \cardin\bigg[\sum_{q=(0,q_s),x,y} \pi(q_s,q)p(x|q_s,q)p(y|x,q_s)\log p(y|x,q_s) + \\
                &\quad\quad\quad \sum_{q=(0,q_s),x,y,s} \pi(s,q)p(x|s,q)p(y|x,s)\log p(y|x,s) \bigg]\\
                &= -\frac{2}{p+3} \sum_{q=(1,q_s),x\neq q_s,y} p(x|q_s,q)p(y|x,q_s)\log p(y|x,q_s) - \\
                &\quad \frac{2p}{p+3}\bigg[\sum_{q=(0,q_s),x,y} \mathds{1}[x=q_s]p(y|x,q_s)\log p(y|x,q_s) + \\
                &\quad \frac{1-p}{(\cardin-1)(p+3)}\sum_{q=(0,q_s),x,y,s} \mathds{1}[x=s]p(y|x,s)\log p(y|x,s) \bigg]\\
                &= -\frac{2}{p+3} \sum_{q=(1,q_s),x\neq q_s,y} p(x|q_s,q)p(y|x,q_s)\log p(y|x,q_s) - \\
                &\quad \frac{2p}{p+3}\bigg[\sum_{y} p(y|q_s,q_s)\log p(y|q_s,q_s) + \\
                &\quad \frac{1-p}{(\cardin-1)(p+3)}\sum_{x,y}p(y|x,x)\log p(y|x,x) \bigg]\\
                &= -\frac{2}{p+3} \sum_{x\neq q_s,y} \frac{1-p}{\cardin-1}p(y|x,q_s)\log p(y|x,q_s)  \\
                &= \frac{2(1-p)}{p+3} \frac{1}{\cardin-1}\sum_{x\neq q_s} 1  \\
                &= \frac{2(1-p)}{p+3} .
\end{align*}
Thus, there exists an input distribution with $P_{Y_i|Y^{i-1}} = T_{Y|Q}$ and the same rate as the upper bound in Lemma \ref{lemma:optimal_value_function_lt8}. This concludes the proof.
\end{proof}

\section{Upper Bound on the Feedback Capacity for $\lvert\cX\lvert > 8$}\label{sec:app-ub-large}
\begin{proof}[Proof of Theorem \ref{thm:ub-large_cardinality}]
The Q-graph obtained from the numerical results of the RL algorithm applied on the Ising channel with $|\cX|=9$ has 12 nodes, where node $Q=1$ denotes that the decoder does not know the channel states. Nodes $Q=2,\dots,10$ correspond to complete knowledge of the channel state and nodes $Q=11,12$ correspond to partial knowledge of the channel state. We denote each node by the tuple $(q,q_s)$ where $q$ indicates the node number and $q_s$ indicates the channel state that this node contains information about.
The corresponding conditional distribution $T \lb y \lvert q \rb$ of the Q-graph is given by 
\begin{equation} \label{eqn:tyq-gt8}
    T\lb y\lvert q \rb =\begin{cases}
    \frac{1}{\cardin}                             &, q=1 \\
    p              								  &, q \neq 1, q_s=y  \\
    \frac{1-p}{\lb \cardin-1 \rb}  &, q\neq 1, q_s\neq y
                        \end{cases}. 
\end{equation}

Next, we verify that the Bellman equation holds for every $(s,q)$. That is, we verify that the right-hand-side maximum of \eqref{eqn:bellman_duality} equals $V(s,q)$. For this purpose, we use the following identity:
\begin{equation}
	\rho = \frac{1}{3}\log\left(\frac{|\cX|(|\cX|-1)}{4p(1-p)}\right).
\end{equation}
\\
\noindent For $q=1, s\neq 6$
\newline
$x=s$
\begin{align}
    V(s,q)  &= -\rho+\sum_{y=x,s} p(y|x,s)\left[\log\left(\frac{p(y|x,s)}{T(y|q)}\right) + V(x,\phi(q,y))\right] \\
            &= -\rho+1 \left[\log\left(\frac{1}{\frac{1}{|\cX|}}\right) + V\left(s,(1,s))\right)\right] \\
            &= -\rho+\log(|\cX|) + 2\rho-\log(|\cX|) = \rho.
\end{align}
$x\neq s$
\begin{align}
    V(s,q)  &= -\rho+\sum_{y=x,s} p(y|x,s)\left[\log\left(\frac{p(y|x,s)}{T(y|q)}\right) + V(x,\phi(q,y))\right] \\
            &= -\rho+\frac{1}{2} \underbrace{\left[\log\left(\frac{\frac{1}{2}}{\frac{1}{|\cX|}}\right) + V(x,(1,s)))\right]}_{y=s} +  \frac{1}{2} \underbrace{\left[\log\left(\frac{\frac{1}{2}}{\frac{1}{|\cX|}}\right) + V(x,(1,x)))\right]}_{y=x}\\
            &= -\rho+\frac{1}{2} \left[2\log(\frac{|\cX|}{2})+ 2\rho +2 - \log(|\cX|)+ 2\rho - \log(|\cX|)\right] = \rho.
\end{align}
\vspace{10pt}
\noindent For $q=1, s=6 $
\newline
$x=s$
\begin{align}
    V(s,q)  &= -\rho+\sum_{y=x,s} p(y|x,s)\left[\log\left(\frac{p(y|x,s)}{T(y|q)}\right) + V(x,\phi(q,y))\right] \\
            &= -\rho+1 \left[\log\left(\frac{1}{\frac{1}{|\cX|}}\right) + V\left(6,(1,6))\right)\right] \\
            &= -\rho+\log(|\cX|) + 2\rho-\log(|\cX|) = \rho.
\end{align}
$x\neq s$
\begin{align}
    V(s,q)  &= -\rho+\sum_{y=x,s} p(y|x,s)\left[\log\left(\frac{p(y|x,s)}{T(y|q)}\right) + V(x,\phi(q,y))\right] \\
            &= -\rho+\frac{1}{2} \underbrace{\left[\log\left(\frac{\frac{1}{2}}{\frac{1}{|\cX|}}\right) + V(x,(1,6)))\right]}_{y=s} +  \frac{1}{2} \underbrace{\left[\log\left(\frac{\frac{1}{2}}{\frac{1}{|\cX|}}\right) + V(x,(1,x)))\right]}_{y=x}\\
            &= -\rho+\frac{1}{2} \left[2\log(\frac{|\cX|}{2})+ 2\rho - \log(|\cX|)+ 2\rho +2- \log(|\cX|)\right] = \rho.
\end{align}
In that case the Bellman equation is satisfied.\\

\noindent For $q\neq 1, q_s = s$:
\newline
$x=s$
\begin{align}
    V(s,q)  &= -\rho+\sum_{y=x,s} p(y|x,s)\left[\log\left(\frac{p(y|x,s)}{T(y|q)}\right) + V(x,\phi(q,y))\right] \\
            &= -\rho+1 \left[\log\left(\frac{1}{p}\right) + V\left(s,(1,s))\right)\right] \\
            &= -\rho+\log(\frac{1}{p}) + \rho  \\
            &= -\rho+\underbrace{\log(\frac{|\cX|}{p})}_{2\rho} - \log(|\cX|)+ \rho =  2\rho - \log(|\cX|).
\end{align}
$x\neq s$
\begin{align}
    V(s,q)  &= -\rho+\sum_{y=x,s} p(y|x,s)\left[\log\left(\frac{p(y|x,s)}{T(y|q)}\right) + V(x,\phi(q,y))\right] \\
            &= -\rho+\frac{1}{2} \underbrace{\left[\log\left(\frac{\frac{1}{2}}{p}\right) + V(x,(q,s)))\right]}_{y=s} +  \frac{1}{2} \underbrace{\left[\log\left(\frac{\frac{1}{2}}{\frac{1-p}{|\cX|-1}}\right) + V(x,(q,x)))\right]}_{y=x}\\
            &= -\rho+\frac{1}{2} \left[\log\left(\frac{(|\cX|-1)}{4p(1-p)}\right) 2\rho - \log(|\cX|)+ \rho \right] \\
            &= -\rho+\frac{1}{2} \left[\underbrace{\log\left(\frac{|\cX|(|\cX|-1)}{4p(1-p)}\right)}_{3\rho}+ 2\rho - 2\log(|\cX|)+ \rho \right]= 2\rho-\log(|\cX|).
\end{align}
In that case the Bellman equation is satisfied.\\
\vspace{10pt}
\noindent For $q\neq 1, q_s \neq s$
\newline
$x=s$
\begin{align}
    V(s,q)  &= -\rho+\sum_{y=x,s} p(y|x,s)\left[\log\left(\frac{p(y|x,s)}{T(y|q)}\right) + V(x,\phi(q,y))\right] \\
            &= -\rho+1 \left[\log\left(\frac{1}{p}\right) + V\left(s,(q,s))\right)\right] \\
            &= -\rho+\underbrace{\log(\frac{1}{p})}_{2\rho-\log(|\cX|)} + \rho = 2\rho-\log(|\cX|).
\end{align}
$x\neq s, x=q_s$
\begin{align}
    V(s,q)  &= -\rho+\sum_{y=x,s} p(y|x,s)\left[\log\left(\frac{p(y|x,s)}{T(y|q)}\right) + V(x,\phi(q,y))\right] \\
            &= -\rho+\frac{1}{2} \underbrace{\left[\log\left(\frac{\frac{1}{2}}{p}\right) + V(x,(q,s)))\right]}_{y=s} +  \frac{1}{2} \underbrace{\left[\log\left(\frac{\frac{1}{2}}{\frac{1-p}{|\cX|-1}}\right) + V(x,(q,x)))\right]}_{y=x}\\
            &= -\rho+\frac{1}{2} \left[\underbrace{\log\left(\frac{(|\cX|-1)}{4p(1-p)}\right)}_{3\rho-\log(|\cX|)} + 2\rho +2 - \log(|\cX|)+ \rho \right] = 2\rho - \log(|\cX|) + 1.
\end{align}
$x\neq s, x\neq q_s$
\begin{align}
    V(s,q)  &= -\rho+\sum_{y=x,s} p(y|x,s)\left[\log\left(\frac{p(y|x,s)}{T(y|q)}\right) + V(x,\phi(q,y))\right] \\
            &= -\rho+\frac{1}{2} \underbrace{\left[\log\left(\frac{\frac{1}{2}}{\frac{1-p}{|\cX|-1}}\right) + V(x,(q,s)))\right]}_{y=s} +  \frac{1}{2} \underbrace{\left[\log\left(\frac{\frac{1}{2}}{\frac{1-p}{|\cX|-1}}\right) + V(x,(q,x)))\right]}_{y=x}\\
            &= -\rho+\frac{1}{2} \left[2\underbrace{\log\left(\frac{(|\cX|-1)}{1-p)}\right)}_{\rho+1} + 2\rho +2 - \log(|\cX|)+ 2\rho - \log(|\cX|) \right] = 2\rho - \log(|\cX|) + 2.
\end{align}
Here, also, the Bellman equation is satisfied. This concludes the proof.
\end{proof}

\end{appendices}
\end{document}